\documentclass[aps,prx,onecolumn,10pt,superscriptaddress,longbibliography,nofootinbib]{revtex4-2}
\usepackage[utf8]{inputenc}               
\usepackage[english]{babel}               
\usepackage{etoolbox}
\makeatletter
\@namedef{ver@babel.sty}{}                 
\makeatother
\usepackage{xspace}

\usepackage{pifont}

\usepackage{siunitx}
\usepackage{mathtools}                    
\usepackage{amsmath,amssymb,amsthm}
\usepackage{braket}                       
\usepackage{physics}                      
\AtBeginDocument{\RenewCommandCopy\qty\SI}
\theoremstyle{plain}
\newtheorem{theorem}{Theorem}
\theoremstyle{definition}
\newtheorem{lemma}{Lemma}

\newenvironment{restatelemma}[1]{%
  \renewcommand{\thelemma}{\ref{#1}}%
  \begin{lemma}%
}{%
  \end{lemma}%
}
\newenvironment{restatetheorem}[1]{%
  \renewcommand{\thetheorem}{\ref{#1}}%
  \begin{theorem}%
}{%
  \end{theorem}%
}

\usepackage{graphicx}
\usepackage{geometry}
\usepackage{booktabs}
\usepackage{array,multirow,adjustbox}
\usepackage{comment}
\usepackage{chemfig}
\usepackage{chemformula}
\newcommand{\CD}{\mathrm{CD}}      
 
\newcommand{\bigO}{\mathcal{O}}

\usepackage[
  colorlinks = False,
  linkcolor  = blue,
  citecolor  = blue,
  urlcolor   = blue,
  filecolor  = blue,
  pdfborder  = {0 0 0}
]{hyperref}
\begin{document}

\title{Oscillator-qubit generalized quantum signal processing for vibronic models: a case study of uracil cation}

\author{Jungsoo Hong}
\affiliation{SKKU Advanced Institute of Nano Technology (SAINT), Sungkyunkwan University, Suwon 16419, Republic of Korea}
\author{Seong Ho Kim}
\affiliation{Department of Chemistry, Ulsan National Institute of Science and Technology, Ulsan 44919, Republic of Korea}
\author{Seung Kyu Min}
\email{skmin@unist.ac.kr}
\affiliation{Department of Chemistry, Ulsan National Institute of Science and Technology, Ulsan 44919, Republic of Korea}
\author{Joonsuk Huh}
\email{joonsukhuh@yonsei.ac.kr}
\affiliation{Department of Chemistry, Yonsei University, Seoul 03722, Republic of Korea}
\affiliation{Department of Quantum Information, Yonsei University, Seoul 03722, Republic of Korea}

\date{\today}

\begin{abstract}
Hybrid oscillator-qubit processors have recently demonstrated high-fidelity control of both continuous- and discrete- variable information processing. However, most of the quantum algorithms remain limited to homogeneous quantum architectures. Here, we present a compiler for hybrid oscillator-qubit processors, implementing state preparation and time evolution. In hybrid oscillator-qubit processors, this compiler invokes generalized quantum signal processing (GQSP) to constructively synthesize arbitrary bosonic phase gates with moderate circuit depth $\bigO(\log(1/\varepsilon))$. The approximation cost is scaled by the Fourier bandwidth of the target bosonic phase, rather than by the degree of nonlinearity. Armed with GQSP, nonadiabatic molecular dynamics can be decomposed with arbitrary-phase potential propagators. Compared to fully discrete encodings, our approach avoids the overhead of truncating continuous variables, showing linear dependence on the number of vibration modes while trading success probability for circuit depth. We validate our method on the uracil cation, a canonical system whose accurate modeling requires anharmonic vibronic models, estimating the cost for state preparation and time evolution.
\end{abstract}

\maketitle

\onecolumngrid 

\section{Introduction} \label{sec:introduction}

Recent experimental progress in both discrete-variable (DV, qubit) and continuous-variable (CV, oscillator) information processing has established the foundation for hybrid CV–DV quantum architectures~\cite{liu_hybrid_2024}. Despite these advancements, applications of hybrid CV-DV processors are still in an early stage of exploration. Conventionally, qubits and oscillators have been employed in complementary roles: qubits provide access to non-Gaussian resources~\cite{nonGaussian2025}, while oscillators serve as quantum buses mediating entanglement between qubits in quantum platforms such as trapped ions~\cite{sorensen2000entanglement} and circuit QED~\cite{Eickbusch2022}. The hybrid CV-DV architecture exploits both continuous and discrete degrees of freedom and extends their capabilities beyond their auxiliary roles by harnessing native representations of variables and utilizing both types of non-classical operations, such as non-Gaussian operations and qubit entanglements with arbitrary phases.

Such hybrid control naturally matches problems where CV and DV degrees of freedom coexist. The simulation of nonadiabatic molecular dynamics, which requires explicit inclusion of electron–nuclear couplings, represents a regime where hybrid oscillator–qubit architectures are expected to outperform homogeneous qubit-based processors and classical methods~\cite{MacDonell2021,Kang2024}. Simulations of such hybrid CV–DV systems on qubit-only hardware incur additional overheads to represent high-dimensional bosonic operators with binary qubits. Although a single truncated oscillator uses only logarithmic number of qubits, the dominant overheads lie in arithmetic and phase-synthesis steps. From the classical perspectives, multiconfiguration time-dependent Hartree (MCTDH)~\cite{Worth2008} scales exponentially with the number of modes, and linear-vibronic-coupling (LVC) analog quantum simulators~\cite{MacDonell2021, Valahu2023, Whitlow2023, wang2023_CI_cQED} are effectively near-harmonic, failing to capture essential anharmonic effects. This motivates the development of a more general compiler for nonlinear bosonic dynamics on hybrid CV–DV hardware.

In this paper, we introduce oscillator–qubit generalized quantum signal processing (OQ-GQSP), a method for synthesizing arbitrary bosonic phase gates. This approach overcomes the challenges of implementing non-Gaussian operations through a constructive framework that enables robust and analytic determination of gate parameters. OQ-GQSP invokes generalized quantum signal processing (GQSP)~\cite{Motlagh2024} to oscillator–qubit architectures, allowing hybrid processors to implement arbitrary analytic potential propagators conditioned on electronic states—an essential capability for simulating nonadiabatic molecular dynamics. We apply this method to the uracil cation~\cite{MATSIKA2008356}, whose ultrafast relaxation through conical intersections (CIs) requires anharmonic models to accurately capture nonadiabatic dynamics.

Because OQ-GQSP inherits the quantum signal processing structure, it supports both coherent~\cite{vasconcelos2025BE} and incoherent accumulation schemes for composing multiple OQ-GQSP blocks. In this work, we employ incoherent accumulation, which reduces circuit depth at the cost of postselection. The success probability of postselection increases with circuit depth, enabling a quantitative trade-off analysis for optimizing resource budgets on a given hardware specification.

The hybrid oscillator-qubit quantum architecture organizes the solution as five steps:
(i) the application layer defines ultrafast relaxation dynamics as the target problem;
(ii) an algorithmic layer realizes this as time evolution under the anharmonic vibronic model;
(iii) a compiler layer synthesises the bosonic nonlinear phase gates needed by the anharmonic vibronic model;
(iv) the instruction set architecture (ISA) enumerates the native oscillator–qubit operations~\cite{Araz2025};
(v) finally, the hardware layer comprises trapped ions and circuit-QED that implement those primitives.
Each lower layer exists only to support the one above it, guaranteeing that every step of the stack is both necessary for chemical accuracy and sufficiently native for quantum hardware. Figure~\ref{fig:method} provides a visual summary of our workflow, including these layers. We encode the four electronic states of uracil cation in a qubit-encoding representation, map off-diagonal vibronic couplings to multi-controlled displacement (MCD) gates, and implement diagonal anharmonic potentials through OQ-GQSP. The complete circuit evolves via Trotterized time evolution, directly incorporating both linear couplings and nonlinear anharmonic potential effects at each timestep.

\begin{figure*}
    \includegraphics[width=0.95\linewidth]{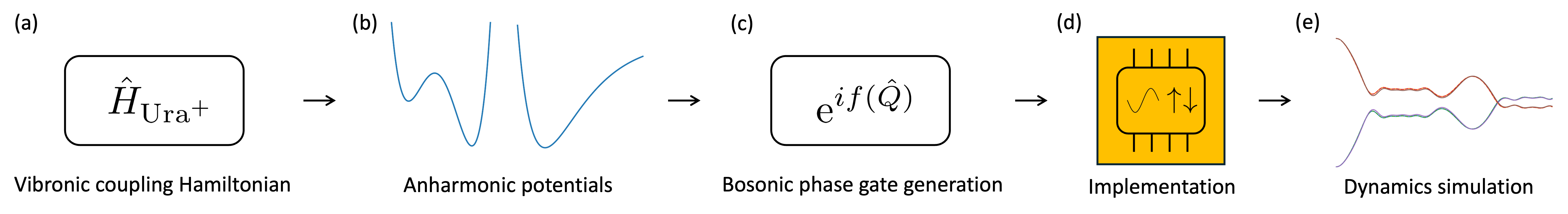}
     \caption{\label{fig:method} Quantum simulation workflow for nonadiabatic dynamics of the uracil cation. (a) Target system specified by the vibronic coupling Hamiltonian $\hat H_{\mathrm{Ura}^+}$. (b) Anharmonic potential energy surfaces $V(x)$ beyond the harmonic approximation. (c) Construction of nonlinear bosonic phase gates via oscillator–qubit generalized quantum signal processing (OQ–GQSP). (d) Implementation on a hybrid oscillator–qubit quantum processor. (e) Measurements yield nonadiabatic electronic population dynamics.}
\end{figure*}

The remainder of this paper is organized as follows. The theory section presents the theoretical framework for anharmonic vibronic coupling models and introduces the uracil cation Hamiltonian as our example system. The method section details compilation stages including the ISA of oscillator-qubit systems and our OQ-GQSP algorithm for implementing anharmonic potentials. The numerical experiment section provides numerical demonstrations of the potential energy surface reconstruction and comparative dynamics simulations. In the resource estimation section, we compare the computing resources and time complexity for simulating the vibronic coupling model. In the discussion and conclusions section, we discuss the challenges of the OQ-GQSP method and compare it with the other methods.  
\section{Theory}\label{sec:theory}

\subsection{Vibronic coupling Hamiltonian}
For a molecule with $N$ electronic states and $M$ mass-frequency scaled normal mode coordinates $\{\hat{Q}_r\}$ ($r\in\{0,...,M-1\}$) with $\hat{P}_r\equiv -i\frac{\partial}{\partial Q_r}$, the vibronic coupling Hamiltonian takes the form:
\begin{align}\label{eq:VC_general}
    \begin{aligned}
\hat{H} = \, & \mathbb{I}_{\mathrm{el}}\otimes (\hat{T}+\hat{V}_0) + \hat{W}\\
=\, &\mathbb{I}_{\mathrm{el}} \otimes \sum_{r=0}^{M-1} \frac{\omega_r}{2}(
\hat{P}_r^2+\hat{Q}_r^2) + \sum_{n,m=0}^{N-1}\ket{n}\bra{m} \otimes h_{nm}(\hat{Q}_0,\dots,\hat{Q}_{M-1}),
    \end{aligned}
\end{align}
where $\hat{T}$ and $\hat{V}_0$ denote the kinetic and potential energy operators, respectively, and $\hat{W}$ is to describe additional change in potentials with respect to $\hat{V}_0$. The momentum and position operators can be defined as $\hat{P}_r = (\hat{a}_r-\hat{a}_r^{\dagger})/\sqrt{2}$ and $\hat{Q}_r = (\hat{a}_r+\hat{a}_r^{\dagger})/\sqrt{2}$, where $\hat{a}_r$ and $\hat{a}^{\dagger}$ are the bosonic annihilation and creation operators, respectively with $\omega_r$ denoting the harmonic frequency of mode $r$. The coupling functions $h_{nm}(\{\hat{Q}_r\})$ depend on the position operators and determine both the molecular physics and the computational complexity of the dynamics simulation which can be expanded as $h_{nn}(\{\hat{Q}_r\}) = E_n+\sum_r \kappa_r^{(n)}\hat{Q}_r + \sum_{r,s} \gamma_{rs}^{(n)}\hat{Q}_r\hat{Q}_s+\cdots$ for diagonals and $h_{nm}(\{\hat{Q}_r\}) = \sum_r\lambda_r^{(nm)}\hat{Q}_r+\sum_{r,s}\mu_{rs}^{(nm)}\hat{Q}_r\hat{Q}_s+\cdots$ for off-diagonals ($n\neq m$). Here, ionization energies $E_n$, and coefficients ($\kappa$, $\gamma$, $\lambda$, and $\mu$) can be determined as parameters at the reference geometries. Conventionally, the ionization energies $E_n$ can be included in the reference Hamiltonian $\hat{H}_0$ as $\hat{H}_0 = \mathbb{I}_{\mathrm{el}}\otimes (\hat{T}+ \hat{V}_0)+\sum_n \ket{n}\bra{n}E_n$ while we consider $\hat{W}$ as an expansion from the first order.

\subsection{Hamiltonian of uracil cation}
In this work, we adopt the model Hamiltonian of uracil cation from Ref.~\cite{Assmann2015}. The uracil cation consists of four electronic states $\{\ket{D_0},\dots,\ket{D_3}\}$ coupled through multiple vibrational modes (see supporting information in Ref.~\cite{Assmann2015} for details). Based on the photoelectron (PE) spectra, the consideration of symmetry of electronic states and normal modes as well as the magnitude analysis of parameters, the resulting model Hamiltonian can be written as 
\begin{equation}
\hat H_{\text{Ura}^{+}}=
    \hat{H_0}
    +\sum_{n=0}^{N-1}\ket n\bra n  \otimes  \sum_{r=0}^{M-1}f_{r}^{(n)}(\hat Q_r)
    +\sum_{\substack{n,m=0 \\ n \ne m}}^{N-1} \ket n\bra m  \otimes  \sum_{r=0}^{M-1}\lambda^{(nm)}_{r}\hat{Q}_r,
\label{eq:UraHam}
\end{equation} 
where $f_r^{(n)}(\hat{Q}_r)$ is a single-mode function neglecting any mode-mode couplings for diagonal elements. Once we adopt $f_r^{(n)}(\hat{Q}_r) = \kappa_r^{(n)}\hat{Q}_r+\gamma^{(n)}_r\hat{Q}_r^2$, it is called quadratic vibronic coupling (QVC) model. However, the limitations of the QVC approximation are numerically demonstrated in Appendix B. To accurately reproduce the photoelectron (PE) spectra, higher-order anharmonic terms such as $k_r^{(n)}\hat{Q}_r^4$ and Morse potentials $V_r^{(n)}(\hat{Q}_r)$ are introduced, replacing the harmonic potential $\tfrac{\omega_r}{2}\hat{Q}_r^2$ in $\hat{H}_0$. Figure~\ref{fig:Hamiltonian} shows the form of model Hamiltonian without $\hat{H}_0$ coupled to important normal mode coordinates. For uracil cation, Morse functions for bond stretching modes ($\nu_{24}$, $\nu_{25}$, $\nu_{26}$) and quartic functions for out-of-plane deformations ($\nu_{10}$, $\nu_{12}$) are introduced. We note that the off-diagonal elements between $\ket{D_0}$ and $\ket{D_3}$ or between $\ket{D_2}$ and $\ket{D_3}$ vanish. The specific parameter values are tabulated in Ref.~\cite{Assmann2015,Vindel2022} and also given in Appendix A.

\begin{figure}[t]
\includegraphics[width=0.95\linewidth]{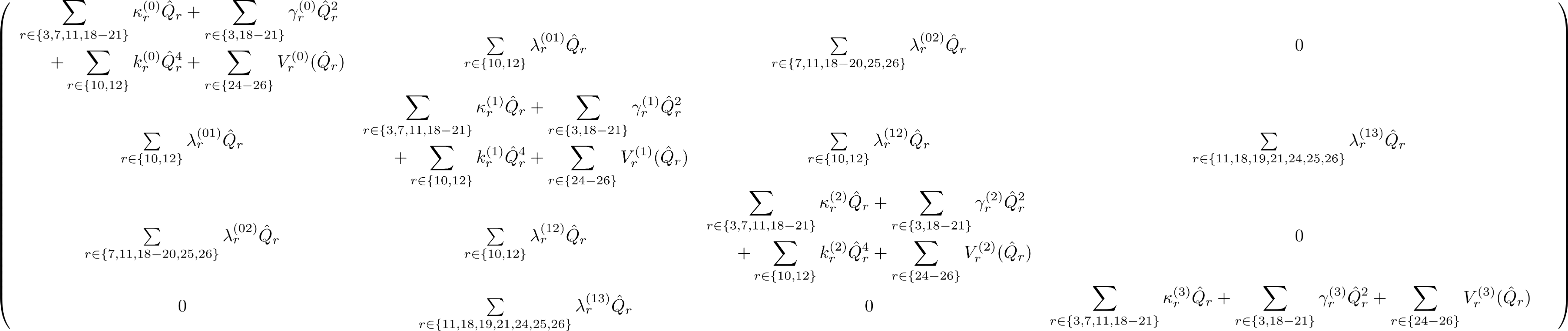}
\caption{\label{fig:Hamiltonian} Four-state vibronic coupling model Hamiltonian for uracil cation. $V_r^{(\alpha)}(\hat{Q}_r) = d_r^{(\alpha)}[e^{(a_r^{(\alpha)})(\hat{Q}_r-q_{r,0}^{(\alpha)})}-1]^2+e_{r}^{(\alpha)}$ is the Morse potential.}
\end{figure}
\section{Method}\label{sec:method}

We now present our method for simulating the uracil cation Hamiltonian on oscillator-qubit quantum processors. Our approach exploits the natural correspondence between molecular degrees of freedom and quantum hardware where the electronic states map to qubits while vibrational modes map to quantum oscillators. The key challenge lies in implementing the anharmonic diagonal potentials that go beyond the native instruction set of current hybrid oscillator-qubit quantum platforms.

Modern trapped-ion and circuit QED platforms provide an ISA that includes conditional displacements, conditional phase space rotations, and multi-qubit gates~\cite{liu_hybrid_2024,Araz2025}. Our compilation strategy systematically translates the molecular Hamiltonian into sequences of these hardware-native operations. The compilation proceeds in three stages: (1) encoding electronic states using an inverted unary representation that facilitates multi-level control operations, (2) implementing off-diagonal linear couplings through MCD gates, and (3) realizing diagonal anharmonic potentials via OQ-GQSP that extends the capabilities of the basic instruction set. 

\subsection{Electronic state encoding}

For the $N$-level electronic system, we employ inverted unary encoding using $N$ qubits, where each computational basis state has exactly one qubit in the $\ket{0}$ state. For the uracil cation's four electronic states, the mapping is:
\begin{align}
\ket{D_0} \mapsto \ket{1110}, \quad \ket{D_1} \mapsto \ket{1101}, \quad \ket{D_2} \mapsto \ket{1011}, \quad
\ket{D_3} \mapsto \ket{0111}.
\end{align}
This encoding offers two critical advantages over binary representations. First, it enables direct implementation of electronic transitions $\ket{D_n}\bra{D_m}$ using two-body operations on qubits $n$ and $m$, minimizing the additional complex multi-qubit gates required in binary encoding. Second, the redundancy of unary basis states reduces unwanted interference when applying OQ-GQSP to different electronic subspaces—a crucial requirement for implementing block-diagonal anharmonic potentials.

\subsection{Compilation of off-diagonal couplings}

The off-diagonal terms in the vibronic Hamiltonian encode the coupling between electronic states induced by the nuclear motion, driving nonadiabatic dynamics through CIs. These LVCs $\ket{D_n}\bra{D_m} \otimes \hat{Q}_r$ translate naturally to MCD operations in oscillator-qubit architectures~\cite{ha2025}.

The ISA of oscillator-qubit platforms provides the building blocks for constructing MCD gates~\cite{liu_hybrid_2024}. The conditional displacement (CD) gate $\CD_{q,r}(\theta) := \ket{0}_q\!\bra{0}\otimes \mathrm{e}^{+i\theta\hat{Q}_r}+\ket{1}_q\!\bra{1}\otimes \mathrm{e}^{-i\theta\hat{Q}_r}$ implements spin-dependent position displacements and serves as the primary gate throughout the compilation stages. In trapped ions, this is realized through red and blue sideband transitions that produce state-dependent forces~\cite{katz2023}. Circuit QED achieves the same operation through echoed conditional displacement (ECD) sequences with dispersive coupling between transmon qubits and microwave cavities~\cite{Eickbusch2022}. Additionally, two-qubit gates and single qubit Clifford gates complete our native instruction set for constructing MCD operations. For example, the vibronic interaction between $\ket{D}_0$ and $\ket{D_1}$ can be decomposed as:
\begin{align}
\begin{aligned}
\mathrm{e}^{\mathrm{i}\hat{H}_{\text{off}}^{(01)}\Delta t} \ket{\psi} = & \, \mathrm{e}^{i\theta(\ket{D_0}\bra{D_1}+\ket{D_1}\bra{D_0})\hat{Q}}\ket{\psi} \\ 
=& \, \mathrm{e}^{\mathrm{i} \theta(\ket{1110}\bra{1101}+\ket{1101}\bra{1110})\hat{Q}}\ket{\psi}\\
=& \, \mathrm{e}^{\mathrm{i} \theta(\hat{\sigma}_{x}^{(1)}\hat{\sigma}_x^{(0)} + \hat{\sigma}_y^{(1)}\hat{\sigma}_y^{(0)})\hat{Q}}\ket{\psi}\\
=& \, \mathrm{e}^{\mathrm{i} \theta\hat{\sigma}_{x}^{(1)}\hat{\sigma}_x^{(0)}\hat{Q}} \mathrm{e}^{i\theta\hat{\sigma}_y^{(1)}\hat{\sigma}_y^{(0)}\hat{Q}}\ket{\psi},
\end{aligned}
\end{align}
where $\ket{\psi}$ is the state in the inverted unary encoded space and $\hat{\sigma}^{(i)}_{k},\quad( k\in\{x,y,z\}\text{ and } i \in \{0,1,2,3\} )$ is a Pauli gate on the $i$ th qubit. The inverted unary encoding simplifies these implementations: the MCD gate for transition $n \leftrightarrow m$ requires only controlled operations on qubits $n$ and $m$, resulting in constant gate depth independent of the number of electronic states.

This set of operations is capable of directly implementing nonadiabatic dynamics of multilevel linear couplings and harmonic potentials, while this instruction set still cannot implement anharmonic terms that are essential for accurate molecular dynamics.

\subsection{Compilation of anharmonic diagonal potential}

The primary challenge in simulating the nonadiabatic dynamics of the uracil cation on an oscillator–qubit processor lies in implementing physics-inspired potentials, such as the Morse potential. While polynomial potentials can, in principle, be constructed via conventional methods using Baker–Campbell–Hausdorff (BCH) expansions, the overhead and inefficiency of synthesizing each term motivate us to explore alternative approaches. Moreover, the diversity of physically motivated functional forms necessitates a fundamentally different compilation strategy—one that can systematically and uniformly represent arbitrary potentials. To this end, we aim to construct a nonlinear bosonic phase gate that approximates arbitrary functions within a reliable error bound and remains robust under arbitrary quantum states, thereby enabling its direct use as a gate in quantum simulations.

The key insight is that a sequence of noncommuting gates can effectively accumulate nonlinear effects~\cite{Low2016}. Rather than explicitly reconstructing interaction terms order by order and canceling unwanted contributions via BCH expansions, Park et al.~\cite{Park2024} employed a direct Fourier series approximation with iterative noncommuting Rabi gates. Using this method, for example, a quartic potential term such as $\exp(\mathrm{i}0.2 \hat{Q}^4)$ can be approximated with ten Rabi gates, achieving a fidelity of $0.9932$~\cite{Park2024}.

A critical caveat in applying such gates to quantum dynamics is the vacuum-state dependence of the synthesized nonlinear bosonic phase gates. This dependence can introduce systematic errors when acting on general time-evolved states beyond the initial vacuum state. Notably, the iterative structure of Rabi gate sequences~\cite{Park2024} for generating Fourier series functions is formally analogous to quantum signal processing (QSP) applied to the oscillator–qubit hybrid system. QSP can represent a continuous function of the matrix with non-commuting sequences where angle parameters can be efficiently calculated~\cite{laneve2025generalizedquantumsignalprocessing,ni2025inversenonlinearfastfourier}. As typically used with QSP, Jacobi-Anger expansions with Chebyshev polynomial may offer improved robustness and convergence then Fourier series approximations. However, their implementation requires direct block encodings of bosonic operators such as $\begin{bmatrix} \hat{Q} & * \\ * & * \end{bmatrix}$—which, to the best of our knowledge, remains infeasible on hybrid oscillator–qubit architectures. In particular, Ref.~\cite{liu2024toward} demonstrates a QSP implementation on hybrid architectures by block-encoding bosonic phase operators using the $\CD$ gate, which acts as a Z-basis signal operator.

Within the conventional QSP framework, one aims to represent $\mathrm{e}^{-\mathrm{i}x\tau}$ approximately by a polynomial in the input variable $x$, later promoted to the operator or Hamiltonian, with $\tau$ denoting a real parameter corresponding to the effective simulation time~\cite{martyn2023efficient}. However, as emphasized in Ref.~\cite{martyn2023efficient}, this approach encounters fundamental limitations. The standard QSP constructs real polynomials with fixed parity in default. Extending the class of functions that can be approximated beyond even or odd requires additional circuit depth of linear combination of unitaries.

GQSP~\cite{Motlagh2024} lifts these limitations by generalizing the signal operator from $U(1)$ rotations to $\mathrm{SU}(2)$ rotations, and by introducing an asymmetric signal operator of the form $\hat{A}_{\hat{U}} = \begin{bmatrix} \hat{U} & 0 \\ 0 & I \end{bmatrix}$. This generalization eliminates the parity constraints and permits the construction of polynomials with arbitrary complex coefficients. The theorem of GQSP is often stated as follows for the case of including positive and negative powers of a Laurent polynomial~\cite{Motlagh2024,yamamoto2024robust,laneve2025generalizedquantumsignalprocessing}:
\begin{theorem}[Generalized Quantum Signal Processing (GQSP)~\cite{Motlagh2024}]\label{thm:GQSP}
Let $\boldsymbol\phi=(\phi_{-d},\dots,\phi_{d})$ and $\boldsymbol\theta=(\theta_{-d},\dots,\theta_{d})$. For any $d \in \mathbb{N}$, there exist parameters $\boldsymbol{\theta}, \boldsymbol{\phi} \in \mathbb{R}^{2d+1}$ and $\lambda \in \mathbb{R}$ such that:
\begin{align}
\mathrm{e}^{i\lambda\hat{\sigma}_z}\mathrm{e}^{i\phi_{-d}\hat{\sigma}_x} \mathrm{e}^{i\theta_{-d}\hat{\sigma}_z}\left( \prod_{r=-d+1}^{0} \hat{B} \mathrm{e}^{i\phi_{r}\hat{\sigma}_x} \mathrm{e}^{i\theta_{r}\hat{\sigma}_z}\right)\left( \prod_{r=1}^{d} \hat{A} \mathrm{e}^{i\phi_{r}\hat{\sigma}_x} \mathrm{e}^{i\theta_{r}\hat{\sigma}_z} \right) = 
\begin{bmatrix}
  F(x) &-(G(x))^{*} \\
  G(x) & (F(x))^{*} 
\end{bmatrix},\,
\end{align}
for all $x \in \mathbb{T}$, where $\mathbb{T} = \{x \in \mathbb{C} : |x| = 1\}$ and
\begin{align}
\begin{aligned}
 \hat{A} &= (\lvert 0 \rangle_q \langle 0 \rvert \otimes x ) + (\lvert 1 \rangle_q \langle 1 \rvert \otimes \mathbb{I})\\ &= 
        \begin{bmatrix}
        x & 0 \\
        0 & 1\\
        \end{bmatrix}, \\
 \hat{B} &= (\lvert 0 \rangle_q \langle 0 \rvert \otimes \mathbb{I} ) + (\lvert 1 \rangle_q \langle 1 \rvert \otimes x^{-1} )\\ &= 
        \begin{bmatrix}
        1 & 0 \\
        0 & x^{-1}\\
        \end{bmatrix},\\ 
\end{aligned}
\end{align}
if and only if:
\begin{enumerate}
    \item $F, G \in \mathbb{C}[x]$ with $|\deg(F)|, |\deg(G)| \leq d$.
    \item $|F(x)|^2 + |G(x)|^2 = 1$.
\end{enumerate}
\end{theorem}

In Ref.~\cite{Motlagh2024}, generating nonlinear phase functions via GQSP proceeds through a two-step approximation strategy. First, the target potential $V(x)$ is approximated by a truncated Fourier series. Second, each exponential term of the form $\exp(\mathrm{i} a_n \cos(n x))$ or $\exp(\mathrm{i} b_n \sin(n x))$ is further expanded using the Jacobi–Anger identity:
\begin{align}
\begin{aligned}
\exp\left(\mathrm{i} V(\hat{Q}) \Delta t \right)
&\approx \exp\left(\mathrm{i} \Delta t \sum_n \bigl[ a_n \cos(n\hat{Q}) + b_n \sin(n\hat{Q}) \bigr] \right) \\
&\approx \prod_k \left( \sum_n \mathrm{i}^n J_n(a_k \Delta t) \mathrm{e}^{\mathrm{i} n k \xi \hat{Q}} \right)
             \left( \sum_n J_n(b_k \Delta t) \mathrm{e}^{\mathrm{i} n k \xi \hat{Q}} \right),
\end{aligned}
\end{align}
where the Jacobi–Anger expansions are
\begin{align}
\mathrm{e}^{\mathrm{i} t \cos \theta} = \sum_{n = -\infty}^{\infty} \mathrm{i}^n J_n(t) \mathrm{e}^{\mathrm{i} n \theta}, \qquad
\mathrm{e}^{\mathrm{i} t \sin \theta} = \sum_{n = -\infty}^{\infty} J_n(t) \mathrm{e}^{\mathrm{i} n \theta},
\end{align}
and $J_n$ denotes the $n$-th Bessel function of the first kind.  
This indirect construction incurs additional overhead due to the nested use of both Fourier and Jacobi–Anger expansions.

As an alternative, we adopt a direct Fourier-based approach in which the nonlinear bosonic phase gate $\exp(\mathrm{i} V(\hat{Q}) \Delta t)$ is expressed as a Laurent polynomial:
\begin{align}
\exp(\mathrm{i} V(\hat{Q}) \Delta t) \approx \sum_{n = -d}^{d} c_n \left( \mathrm{e}^{\mathrm{i} \frac{2\pi}{L} \hat{Q}} \right)^n,
\end{align}
where $c_n \in \mathbb{C}$ are Fourier coefficients and $V$ is an analytic potential function.  
This formulation is natively compatible with GQSP and avoids the compounded approximation cost associated with the Jacobi–Anger expansion.

\begin{figure} 
\includegraphics[width=0.95\linewidth]{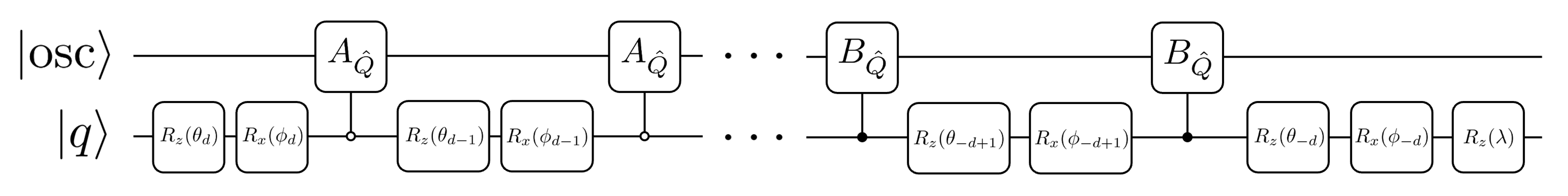}
\caption{\label{fig:OQ-GQSP}  Quantum circuit for OQ-GQSP. $\ket{\mathrm{osc}}$ is a quantum oscillator and $\ket{q}$ is a qubit register. The left half of the circuit generates positive powers, and the right half generates negative powers of the Laurent (Fourier) series terms.}
\end{figure}

Applying GQSP Theorem~\ref{thm:GQSP} to oscillator-qubit processors towards constructing state-dependent nonlinear bosonic phase gates through Fourier series approximation is established with the following Lemmas:
\begin{lemma}[Oscillator-Qubit Generalized Quantum Signal Processing (OQ-GQSP)]\label{lem:GQSP_qubit_oscillator_main}
Given a Z-basis controlled position displacement gate, one can construct complementary generalized signal operators
\begin{align}
\begin{aligned}
 A_{\hat{Q}} &= (\lvert 0 \rangle_q \langle 0 \rvert \otimes \hat{U} ) + (\lvert 1 \rangle_q \langle 1 \rvert \otimes \mathbb{I}_{\mathrm{osc}})\\ &= 
        \begin{bmatrix}
        \hat{U} & 0 \\
        0 & \mathbb{I}_{\mathrm{osc}}\\
        \end{bmatrix}, \\
 B_{\hat{Q}} &= (\lvert 0 \rangle_q \langle 0 \rvert \otimes \mathbb{I}_{\mathrm{osc}} ) + (\lvert 1 \rangle_q \langle 1 \rvert \otimes \hat{U}^{\dagger} )\\ &= 
        \begin{bmatrix}
        \mathbb{I}_{\mathrm{osc}}& 0 \\
        0 & \hat{U}^{\dagger}\\
        \end{bmatrix},\\ 
\end{aligned}
\end{align}
with one reusable ancilla qubit and 2 Z-basis CD gates, where $\hat{U} = \mathrm{e}^{\mathrm{i}\pi/L\hat{Q}}$. With these two signal operators, one can construct an arbitrary Fourier series $F(\hat{U})$ of the truncation order $d$ with the period $2L$ of the bosonic phase operator $\mathrm{e}^{\mathrm{i}\pi/L\hat{Q}}$ by Theorem~\ref{thm:GQSP}:
\begin{align}
\begin{aligned}
\operatorname{OQ-GQSP}:=&\mathrm{e}^{i\lambda\hat{\sigma}_z}\mathrm{e}^{i\phi_{-d}\hat{\sigma}_x} \mathrm{e}^{i\theta_{-d}\hat{\sigma}_z}\left( \prod_{r=-d+1}^{0} \hat{B}_{\hat{Q}} \mathrm{e}^{i\phi_{r}\hat{\sigma}_x} \mathrm{e}^{i\theta_{r}\hat{\sigma}_z}\right)\left( \prod_{r=1}^{d} \hat{A}_{\hat{Q}} \mathrm{e}^{i\phi_{r}\hat{\sigma}_x} \mathrm{e}^{i\theta_{r}\hat{\sigma}_z} \right) \\
=&\begin{pmatrix}
  F(\hat{U}) &-(G(\hat{U}))^{\dagger} \\
  G(\hat{U}) & (F(\hat{U}))^{\dagger} 
\end{pmatrix},
\end{aligned}
\end{align}
where $F(\hat{U}) = \sum\limits_{n=-d}^{d} c_n \mathrm{e}^{\mathrm{i}n\pi/L\hat{Q}}$ and $c_n$ is the Fourier series coefficient.
\end{lemma}

The proof follows directly from the symmetry of the Z-basis displacement operator shown in Appendix C. OQ-GQSP can generate arbitrary bosonic phase functions with Fourier series approximation, which is possible by removing the parity conditions of the previous hybrid oscillator-qubit Z-basis QSP in Ref~\cite{liu2024toward}. Armed with Lemma~\ref{lem:GQSP_qubit_oscillator_main}, state-dependent nonlinear bosonic phase gates can be implemented on a hybrid oscillator-qubit processor described with the abstract quantum circuit in Fig.~\ref{fig:state_dependent_nonlinear_bosonic_phase_gate}.

\begin{figure}
\includegraphics[width=8cm]{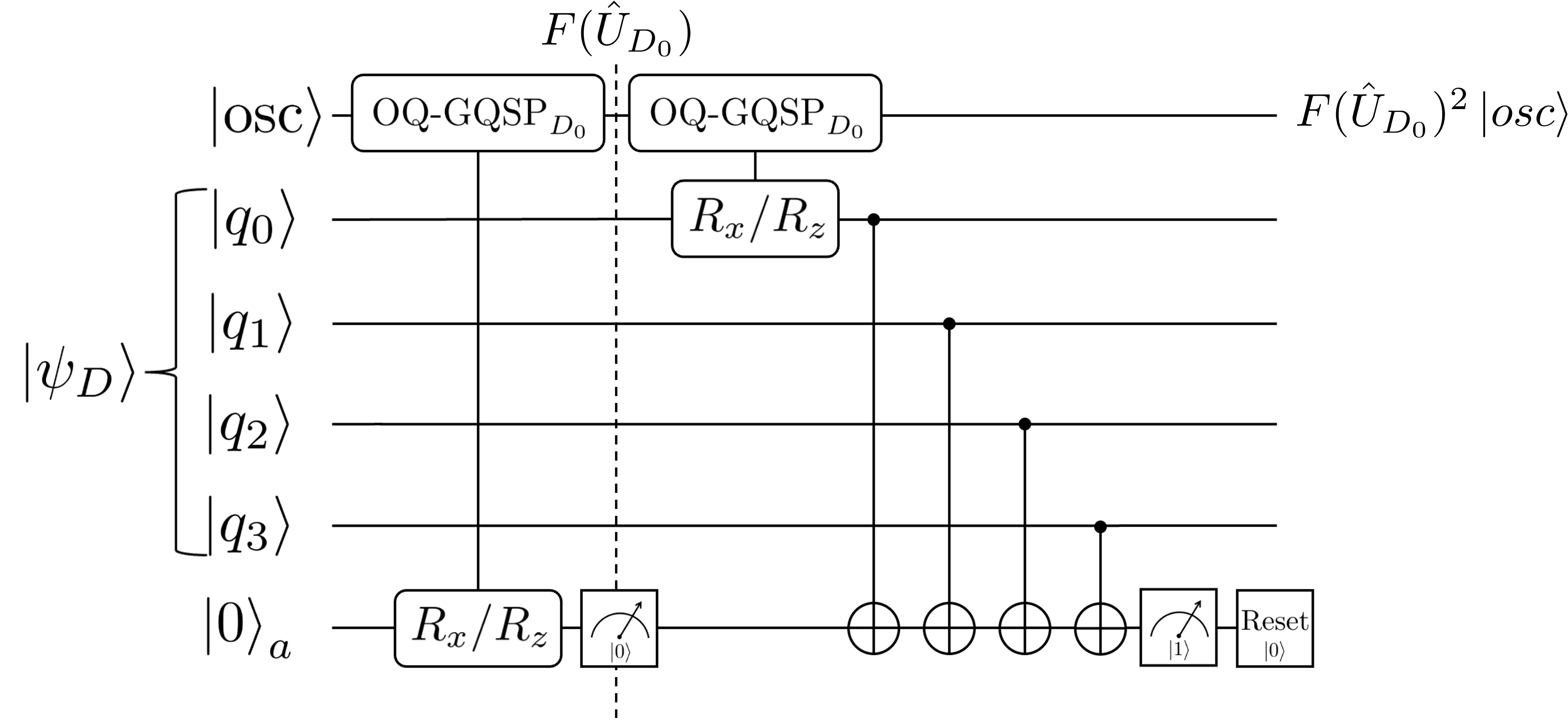}
\caption{\label{fig:state_dependent_nonlinear_bosonic_phase_gate} Quantum circuit for state-dependent nonlinear bosonic phase gate with OQ-GQSP. $\ket{\psi_D}$ is a electronic state encoded with unary code.}
\end{figure}

\begin{lemma}[OQ-GQSP to state-dependent nonlinear bosonic phase gate]\label{lem:GQSP_qubit_oscillator_state_dependent}
Given a state in oscillator-qubit processor $\ket{\psi_D}\ket{\mathrm{osc}}\ket{0}_a$ with inverted-unary-encoded qubits $\ket{\psi_{D}} = a_0\ket{D_0}+a_1\ket{D_1}+a_2\ket{D_2}+a_3\ket{D_3}$, an oscillator $\ket{\mathrm{osc}}$, and an ancilla qubit $\ket{0}_a$, we can implement a state-dependent nonlinear bosonic phase gate
$\mathrm{e}^{i\ket{D_n}\bra{D_n}V_{nr}(\hat{Q}_{r})}$
with two OQ-GQSP operators in Lemma~\ref{lem:GQSP_qubit_oscillator_main} and two measurements with which each success probability is $1-\delta$ where $\delta = \lvert | G(\hat{U})  \rvert |^2 \ll 1.$
\end{lemma}
This lemma completes our ISA for simulating anharmonic vibronic dynamics. The inverted unary encoding enables selective application of anharmonic potentials. The high success probability $(1-\delta)$ with $\delta \ll 1$ makes this approach feasible for dynamics simulations of intermediate size molecules. Figure~\ref{fig:state_dependent_nonlinear_bosonic_phase_gate} depicts the circuit implementation, showing how the ancilla qubit mediates between the electronic state and the oscillator to produce the desired state-dependent evolution. With both off-diagonal couplings via MCD gates and diagonal anharmonic potentials via OQ-GQSP now compilable, we can construct the complete quantum circuit for simulating uracil cation dynamics. 

Based on Lemma~\ref{lem:GQSP_qubit_oscillator_main} and Lemma~\ref{lem:GQSP_qubit_oscillator_state_dependent}, we can write the complexity analysis for implementing a single Trotter step of anharmonic vibronic dynamics with OQ-GQSP. The proof is given in Appendix C.
\begin{theorem}[Anharmonic vibronic dynamics with OQ–GQSP]
\label{thm:GQSP_VCdynamics}
Let $N$ be the number of electronic states, $M$ the number of
vibrational modes, and $M'$ the number of anharmonic modes. For
\begin{align}
  \hat{H}_{\mathrm{VC}}
  =\sum_{n=0}^{N-1}\hat H_{\mathrm{diag}}^{(n)}
    +\!\!\sum_{\substack{n,m=0 \\ n \ne m}}^{N-1} \hat H_{\mathrm{off}}^{(n,m)} ,
\end{align}
set
\begin{align}
  \hat H_{\mathrm{diag}}^{(n)}
  =\sum_{r=0}^{M-1}\ket n\!\bra n\otimes f_{nr}(\hat Q_r),\qquad
  \hat H_{\mathrm{off}}^{(n,m)}
  =(\ket n\!\bra m)
    \otimes\sum_{r=0}^{M-1}\lambda_{nmr}\hat Q_r .
\end{align}

Fix $L>0$. For each pair $(n,r)$ set $g_{nr}(x):=\exp\!\bigl(\mathrm i\,\Delta t f_{nr}(x)\bigr)$. Assume the analyticity and boundedness hypotheses of Lemma~\ref{lem:trunc-L} for each $g_{nr}(x)$. Let $\varepsilon\in(0,1)$ denote the total target error and split it equally between product–formula and Fourier–truncation contributions:
\begin{align}
  \varepsilon=\varepsilon_{\mathrm{Trot}}+\varepsilon_{\mathrm{Fourier}},
  \qquad
  \varepsilon_{\mathrm{Trot}}=\varepsilon_{\mathrm{Fourier}}=\varepsilon/2.
\end{align}

Let $\Delta t\in(0,1]$ and $\varepsilon_{\mathrm{Trot}}\in(0,1)$ and choose $\Delta t$ so that
\begin{align}
  \frac{\Delta t^{2}}{2}\,
  \Bigl\|\Bigl[\sum_{n}\hat H_{\mathrm{diag}}^{(n)}\!,\ \sum_{n \ne m}\hat H_{\mathrm{off}}^{(n,m)}\Bigr]\Bigr\|
  \le \varepsilon_{\mathrm{Trot}} .
\end{align}
Then OQ–GQSP yields a unitary
\begin{align}
  \tilde U
  =\prod_{n=0}^{N-1}\mathrm{e}^{-i\hat H_{\mathrm{diag}}^{(n)}\Delta t}
    \prod_{\substack{n,m=0 \\ n \ne m}}^{N-1}\mathrm{e}^{-i\hat H_{\mathrm{off}}^{(n,m)}\Delta t}
\end{align}
such that $\bigl\|\mathrm{e}^{-i\hat H_{\mathrm{VC}}\Delta_t}-\tilde U\bigr\|
\le \varepsilon$,
using $\mathcal{O}\!\bigl(N M'\ln\varepsilon^{-1} + N^2 M\bigr)$
queries to $\mathrm{CD}$ gates and achieving success probability
$\,(1-\delta)^{2M'N}$. 
\end{theorem}

\subsection{Trotterization and circuit structure}

The time evolution operator is implemented through Trotterization. The first-order Trotter decomposition becomes:
\begin{align}
\mathrm{e}^{-\mathrm{i}\hat{H}_{\text{VC}}\Delta t} \approx  \left(\prod_{r=0}^{M-1}\prod_{n=0}^{N-1} \mathrm{e}^{-\mathrm{i}\hat{H}_{\mathrm{diag}}^{(n,r)}\Delta t}\right)\left(\prod_{r=0}^{M-1}\prod_{\substack{n,m=0 \\ n \ne m}}^{N-1} \mathrm{e}^{-\mathrm{i}\hat{H}_{\mathrm{off-diag}}^{(m,n,r)}\Delta t}\right) + \mathcal{O}(\Delta t^2),
\label{eq:trotter}
\end{align}
where $\hat{H}_{\mathrm{diag}}^{(n,r)}$ contains both harmonic and anharmonic contributions for electronic state $n$ and mode $r$, while $\hat{H}_{\mathrm{off}}^{(n,m,r)}$ represents the LVC between states $n$ and $m$ through mode $r$. Diagonal terms are implemented with OQ-GQSP or standard CD gates for harmonic potentials, and off-diagonal terms with MCD gates.

\begin{figure}
\includegraphics[width=8.0cm]{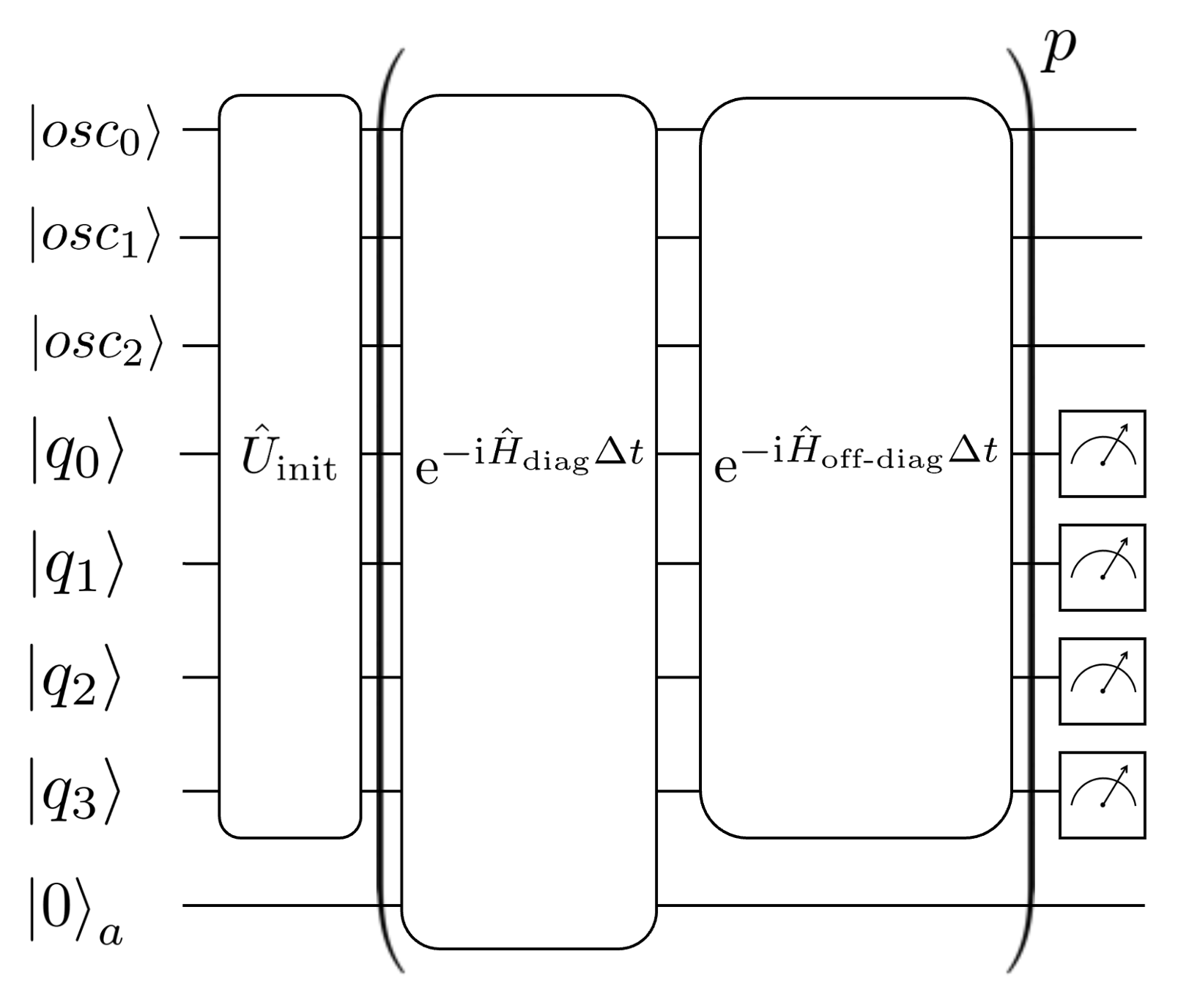}
\caption{\label{fig:quantumcircuit} Abstract quantum circuit for uracil cation Hamiltonian dynamics simulation with $N=4$ electronic states and $M=3$ vibrational modes. The circuit consists of three stages: (i) initialization in a product state, (ii) Trotterized time evolution alternating between off-diagonal couplings and diagonal potentials, and (iii) measurement of electronic populations.}
\end{figure}

For example, to simulate the electronic population dynamics $P_n(t) = \langle\psi(t)|\hat{P}_n|\psi(t)\rangle$ with $\hat{P}_n = |D_n\rangle\langle D_n| \otimes \mathbb{I}_{\mathrm{osc}}$, we evolve the initial state $|\psi_0\rangle$ for time $t$ divided into $p$ Trotter steps as described in Figure~\ref{fig:quantumcircuit} and first‑order product formula in Eq.~\eqref{eq:trotter}. 

We target a total simulation error $\varepsilon$ and split it into a Trotterization and Fourier-truncation errors, $ \varepsilon_{\mathrm{Trot}}=\varepsilon/2,\ \varepsilon_{\mathrm{Fourier}}=\varepsilon/2$.
Let
\begin{align}
\Gamma \;:=\;
\Bigl\|\Bigl[\;\sum_{n,r}\hat H^{(n,r)}_{\mathrm{diag}}\;,\;
             \sum_{r}\sum_{n<m}\hat H^{(n,m,r)}_{\mathrm{off}}\Bigr]\Bigr\|
\;+\;
\sum_{r}\!\!\!\!\sum_{\substack{(n,m)\neq(n',m')\\ n<m,\;n'<m'}}
\bigl\|\,[\hat H^{(n,m,r)}_{\mathrm{off}},\hat H^{(n',m',r)}_{\mathrm{off}}]\,\bigr\|
\end{align}
denote the commutator error bound. With the standard estimate
$\|U(t)-U_{\mathrm{Trot}}(t;p)\|\le \Gamma t^{2}/p$,
we choose the step size and number of steps as
\begin{align}
\Delta t=\frac{\varepsilon_{\mathrm{Trot}}}{\Gamma\, t},
\qquad
p=\Bigl\lceil \frac{\Gamma\, t^{2}}{\varepsilon_{\mathrm{Trot}}}\Bigr\rceil.
\end{align}
The Fourier degree for each diagonal block follows the strip‑analytic bound in Appendix C, $d=\mathcal O\!\big(\ln(p/\varepsilon_{\rm Fourier})\big)$. Per Trotter layer, the diagonal OQ–GQSP synthesis incurs $\bigO(NM'd)$ CD queries, and the off–diagonal couplings contribute $\bigO (N^{2}M)$ MCD queries, for a total of $\mathcal{O}(NM'd+N^{2}M)$ queries. If each heralded primitive fails with probability at most $\delta$, a layer succeeds with probability at least $(1-\delta)^{2NM'}$. Across $p$ layers, the total success probability is $(1-\delta)^{2NM'p}$.

\section{Numerical experiments} \label{sec:numericalexperiments}

We present numerical experiments to give empirical evidence for OQ-GQSP and its application to nonadiabatic dynamics. These experiments complement Lemma~\ref{lem:GQSP_qubit_oscillator_main}, Lemma~\ref{lem:GQSP_qubit_oscillator_state_dependent} and Theorem~\ref{thm:GQSP_VCdynamics} by providing numerical guidance on when OQ‑GQSP attains the targeted accuracy.

\subsection{Nonlinear bosonic phase gate and state preparation}

We first examine the construction of a state-dependent Morse potential for the $\nu_{26}$ carbonyl stretch mode in the $D_2$ electronic state. This mode exhibits a strong anharmonicity in the uracil cation system.

\begin{figure}
\includegraphics[width=0.95\linewidth]{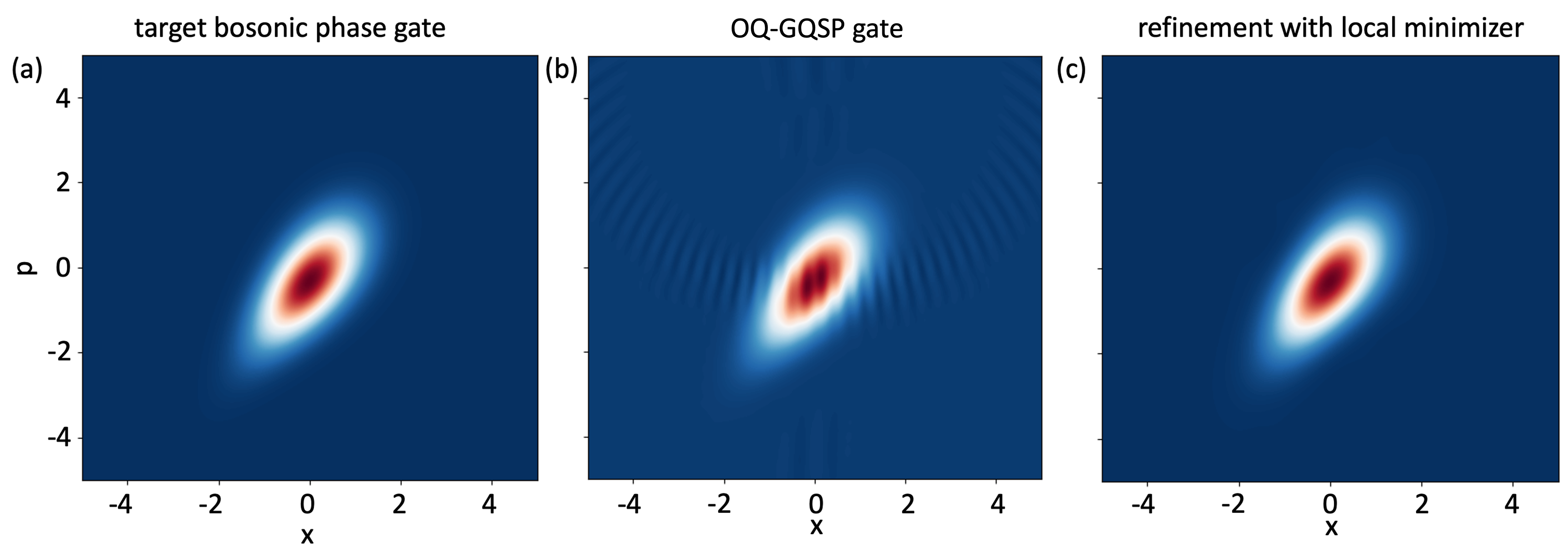}
\caption{\label{fig:Wigner} Wigner function representation of the nonlinear bosonic phase gate for the $\nu_{26}$ mode Morse potential in electronic state $D_2$. (a) Target gate $\mathrm{e}^{-iV_{26}^{(D_2)}(\hat{Q})\Delta t}$ with $\Delta t = 0.3$ fs. (b) OQ-GQSP approximation using $d=39$ Fourier components, achieving fidelity $F=0.9229$. (c) Numerically optimized gate using vacuum-state refinement, reaching $F=0.9999$.}
\end{figure}

Figure~\ref{fig:Wigner} presents the Wigner function representation of the target and approximated gates. The target gate implements the Morse potential $V_{26}^{(D_2)}(\hat{Q})$ for a time step $\Delta t = 0.3$ fs. Using $d=39$ Fourier components, the initial OQ-GQSP approximation achieves a fidelity of $F = |\langle\psi_0| \hat U_{\text{target}}^\dagger \hat U_{\text{OQ-GQSP}}|\psi_0\rangle|^2 = 0.9229$ with the general state. The truncation error $\varepsilon$ can be reduced with increasing the number of Fourier series terms $d=\mathcal{O}(\ln(\varepsilon^{-1}))$, where the proof is given in the appendix C. The characteristic interference pattern visible in Fig.~\ref{fig:Wigner} (b) stems from the finite Fock state truncation with the iterative gate sequence from OQ-GQSP.

To improve the approximation, we employ a local optimization procedure that refines the GQSP angles $\{\theta_j, \phi_j, \lambda \}$ to maximize the fidelity with vacuum-state dependency. This approach works well for dynamics near the reference geometry or initial state preparation before bosonic algorithms such as analog quantum phase estimation~\cite{lin2025analog}. However, general applications with significant vibrational excitation may require the unoptimized gates to maintain robust accuracy across the full phase space. After optimization, the fidelity improves to 0.9999, with the Wigner function becoming visually indistinguishable from the target Fig.~\ref{fig:Wigner} (c).

\subsection{Nonadiabatic dynamics simulation}
To validate the OQ-GQSP method for molecular dynamics, we numerically simulate the time evolution of a reduced 2-mode model of uracil cation including the significant modes in energy spectra $\nu_{21}$ and anharmonic mode $\nu_{26}$ represented with Morse potentials.

\begin{figure} [!b]
\includegraphics[width=0.95\linewidth]{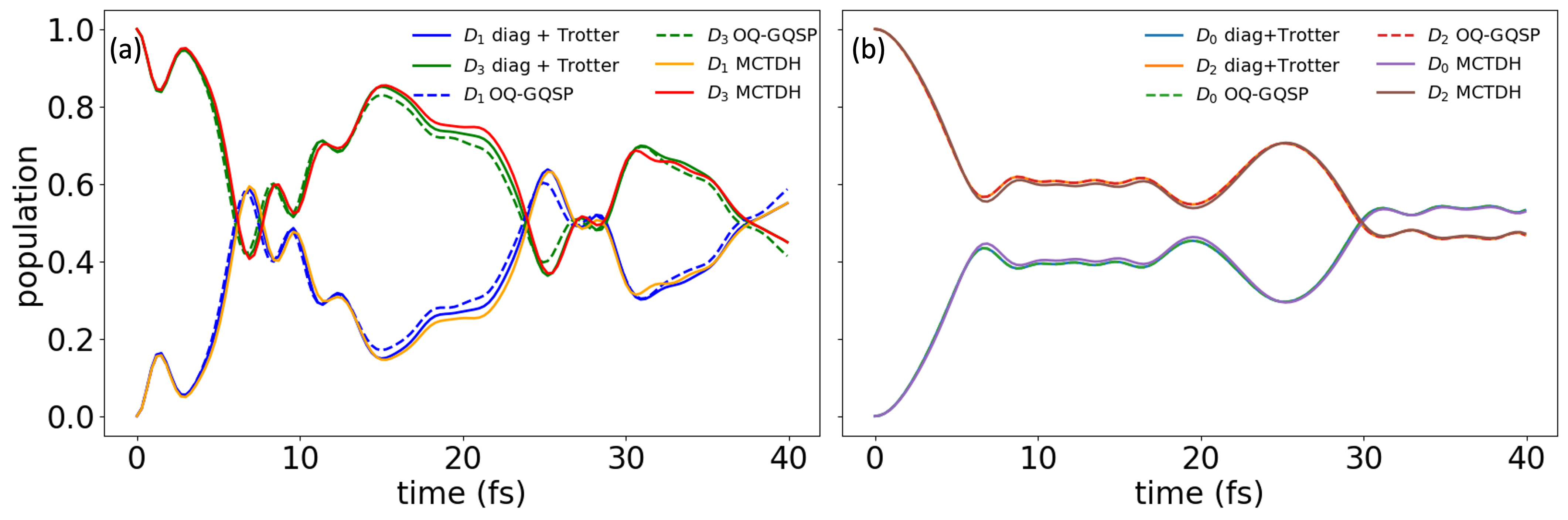}
\caption{\label{fig:OQ-GQSP_dynamics}  Numerical simulation of nonadiabatic dynamics on partial space of uracil cation including the Morse potential. Nonlinear bosonic phase gates are synthesized by OQ-GQSP of nonlinear bosonic phase gate and compared the results of dynamics with replacing OQ-GQSP with numerical diagonalization computing for nonlinear terms. We implemented two conditions of simulation with anharmonic mode $\nu_{26}$: (a) initial state $D_3$ in the subspace $D_1$ and $D_3$ (b) the initial state $D_2$ in the subspace $D_0$ and $D_2$. Dimension of truncated Hilbert space is $30$. Each number of CD gates for OQ-GQSP is 45–100 bounded by the Fourier series approximation error of 0.001. The number of Trotter layers is 100 and the error from the ideal nonlinear bosonic phase gates with matrix diagonalizations is below 0.05.}
\end{figure}

We initialize the system in electronic state $D_3$ and $D_2$ with the anharmonic vibrational mode $\nu_{26}$. The major spectral mode $\nu_{21}$ is displaced by $-4$ in atomic unit, from its ground states, corresponding to vertical excitation from the neutral molecule. Figure~\ref{fig:OQ-GQSP_dynamics} shows the electronic population dynamics over 40 fs, comparing the OQ-GQSP implementation against exact diagonalization in the truncated Hilbert space.

We can observe the population transfer between $\ket{D_0}$ and $\ket{D_2}$ or between $\ket{D_1}$ and $\ket{D_3}$ due to the non-zero off-diaogonal elements in vibronic coupling Hamiltonian in Fig.~\ref{fig:Hamiltonian}. The OQ-GQSP simulation maintains robust fidelity with the exact result, with error within $0.05$ throughout the population dynamics simulation.
\section{Resource estimation}\label{sec:resource estimation}

We analyze the computational resources required for OQ-GQSP and compare them with existing classical and quantum approaches for vibronic dynamics simulation. Our analysis addresses both asymptotic scaling and practical implementation constraints for near-term quantum hardware.

Table~\ref{table:resource estimation} compares the resource requirements across MCTDH for classical method~\cite{Worth2008}, the digital quantum algorithm by Motlagh et al.~\cite{motlagh2024VC}, and our OQ-GQSP method. Each approach targets different regimes and makes trade-offs between memory, runtime, and model expressivity.

\begin{table}
    \centering
    \resizebox{\textwidth}{!}{%
    \begin{tabular}{l|l|l|l}
    Method & Resource & Time complexity & Expressivity \\ 
    \hline
    MCTDH~\cite{Worth2008} & \(n_{\mathrm{SPF}}^p + p n_{\mathrm{SPF}} N_{basis}^{d_{sys}}\)& \(\mathcal{O}(sp^2 {(n_{spf})}^{p+1}+spn_{spf}N_{basis}^{2d_{sys}})\) & \checkmark\\[6pt]
    D. Motlagh, et al.~\cite{motlagh2024VC}& \(\mathcal{O}( (M + d_{\rm{poly}}^2)\log(K) + \log (N))\) qubits & \(\mathcal{O}\left( \left( NM^{d_{\rm{poly}}}(d_{\rm{poly}}\log(K)^2+N) \right)p \right) \)  & \checkmark\\[6pt]
    \textbf{OQ-GQSP} (ours)& $N+1$ qubits, $M$ oscillators & $\mathcal O\!\bigl(\left(
     N M'\ln\varepsilon^{-1}+ N^{2}M\right)p(1-\delta)^{-2M'Np}
   \bigr)$ & $\triangle$\\[6pt]
    \end{tabular}%
    }
    \caption{\label{table:resource estimation}
    Resource scaling comparison for vibronic dynamics simulation methods. Key differences: (1) MCTDH requires exponential memory in system size but no quantum hardware; (2) Motlagh et al. count fault-tolerant Toffoli gates with first order Trotterization for digital quantum computers; (3) OQ-GQSP counts CD gates with first order Trotterization for oscillator-qubit quantum processors. Parameters: $N$ = number of electronic states, $M$ = number of vibrational modes, $n_{\mathrm{SPF}}$ = single-particle functions, $p$ = number of particles, $d_{sys}$ = system coordinates, $N_{basis}$ = primitive basis functions per DOF, $s$ = Hamiltonian terms, $d_{\rm{poly}}$ = polynomial degree, $K$ = grid points, $M'$ = number of anharmonic modes, $\varepsilon$ = target accuracy, $p=\Gamma t^2\varepsilon^{-1}$ number of Trotter layers, $\Gamma$ = commutator bound of Hamiltonian terms, $t$ = simulation time. 
    }
\end{table}

MCTDH exhibits exponential memory scaling with the number of particles and system coordinates, making it infeasible for systems beyond $\sim$10-15 vibrational modes. In contrast, both quantum approaches require only polynomial quantum resources, with OQ-GQSP exploiting native bosonic modes to minimize qubit overhead represented with $\log K$ in Motlagh et al~\cite{motlagh2024VC}. For the cost related to polynomial approximations, $\mathcal{O}(M^{d_{\mathrm{poly}}})$, this factor might become a bottleneck when $M$ or $d_{\mathrm{poly}}$ is large. In practice, $d_{\mathrm{poly}}$ can grow substantially for highly nonlinear potentials or stringent accuracy targets. While OQ-GQSP achieves linear scaling with the number of vibrational modes, the success probability factor $(1-\delta)^{-2M'Np}$ becomes a significant limitation for molecules containing numerous strongly anharmonic modes. For uracil cation consisting of 4 electronic states and among the 12 vibrational modes, there exist quartic and Morse potential modes for $M'=5$. For $M'N = 20$ and target success probabilities $1 - \delta \approx 0.9988$ and $0.9997$, the required number of shots increases by factors of approximately $122$ and $3.14$, respectively, over 100 Trotter layers. The corresponding OQ-GQSP degrees required to synthesize the nonlinear phase gates are $116$ and $261$, respectively. These values are illustrative trade‑off markers under the fully coherent assumption~\cite{laneve2025generalizedquantumsignalprocessing,ni2025inversenonlinearfastfourier}, though they are unlikely to be used directly in near-term simulations. This shows the trade-off between circuit depth and shot overhead: one may reduce the OQ-GQSP degree to shorten circuits at the expense of additional sampling overhead.

\section{Discussion and Conclusions}\label{sec:discussion}

In this work, we developed a hybrid CV-DV quantum compiler that synthesizes arbitrary nonlinear bosonic phase gates with bounded error and applied it to simulate anharmonic vibronic dynamics of the uracil cation. The method preserves vibrational modes in their native bosonic representation while efficiently generating anharmonic potentials via Fourier series approximations. Numerical validation shows that OQ-GQSP captures ultrafast electronic population transfer missed by harmonic approximations. These results indicate that hybrid CV–DV processors can bridge the gap between analog LVC simulators~\cite{MacDonell2021,ha2025} and fully digital vibronic dynamics simulators~\cite{motlagh2024VC}, offering polynomial gate complexity in Fourier bandwidth rather than the degree of polynomial approximation.

While OQ-GQSP represents a step toward practical molecular quantum simulation in oscillator–qubit architectures, several limitations remain. One major challenge is the exponential decay in the probability of success, which scales as $(1 - \delta)^{-2M'Np}$. A fully coherent implementation without the success probability would require ancilla qubits and deeper circuits~\cite{vasconcelos2025BE}. Since oscillator–qubit platforms are unlikely to operate in a fault-tolerant regime, but rather are likely to have advantages in near-term devices~\cite{Kang2024}, we adopted an incoherent protocol with classical overheads instead. Extending the OQ-GQSP to capture mode–mode couplings and reconstruct fully general nonlinear potentials $f_{nm}(\hat{Q}_0, \dots, \hat{Q}_{M-1})$ requires a framework of multivariate QSP (M-QSP). However, efficient angle-finding techniques and a complete characterization of the achievable multivariable polynomials remain open problems, despite recent progress~\cite{Rossi2022multivariable,Mori2024comment,németh2023,ito2024,laneve2025,laneve2025adversaryboundquantumsignal}. From the perspective of quantum machine learning, data re-uploading—a close analogue of M-QSP—is a universal approximator for multivariate functions, though the practicality is an open problem. At least, this theory suggests that general multivariate bosonic operators can, in principle, be numerically synthesized within the M-QSP framework~\cite{quat2025}. On the hardware side, quantum error correction for oscillators remains less developed than for qubits~\cite{Niset2009,Vuillot2019,Noh2020,Wu2021,Wu2023}. More flexible non-stabilizer bosonic codes, tailored to specific tasks, may provide a potential pathway forward. For quantum errors in the near term, recent experiments indicate passive error suppression in QSP‑type sequences~\cite{Bu2025_QSP_trappedions,Zeytinoğlu2024} and an error mitigation technique for hybrid quantum processors is suggested~\cite{park2022}.

We position OQ‑GQSP between fully analog and fully digital approaches to molecular dynamics simulation: a compiler that preserves bosonic modes yet synthesizes anharmonic diagonal potentials with OQ-GQSP on oscillator–qubit ISAs. Recent experiments on trapped-ion and circuit QED platforms~\cite{Whitlow2023,Valahu2023,wang2023_CI_cQED} demonstrated geometric-phase interference around CIs for minimal models, validating the concept of hybrid oscillator-qubit analog simulators~\cite{MacDonell2021}. Theoretical proposals for pre-Born–Oppenheimer (BO) analog simulators~\cite{ha2025} extend this approach beyond the vibronic coupling model for harmonic potential Hamiltonians. Their coupled multi-qubit-boson mapping achieves exponential savings versus classical CASCI-level dynamics but still relies on Trotterized time evolution and continuous analog control of trapped-ion or circuit-QED motional modes, with resource scaling $\bigO (N_o^4M^{d_{\rm{poly}}})$, where $N_o$ is the number of spin-orbitals. One may combine Ref.~\cite{ha2025} and OQ-GQSP methods leading to generalized pre-BO framework to include anharmonic potentials in quantum simulation of coupled electron-nuclear dynamics. The fully digital quantum approach of Motlagh et al.~\cite{motlagh2024VC} achieves quantum simulation for the general vibronic coupling Hamiltonian by digitizing bosonic degrees of freedom. Their singlet fission case study for a 6-state, 21-mode model—requiring 154 qubits and $2.76 \times 10^6$ Toffoli gates for 100 femtoseconds of dynamics. OQ-GQSP offers a constructive middle ground: it preserves bosonic modes in their native Hilbert space while synthesizing anharmonic diagonal potentials through GQSP on oscillator–qubit ISAs. This hybrid strategy directly compiles analytic potentials into hardware-native gates, establishing a general instruction-set architecture in post-NISQ or pre-fault-tolerant trapped-ion and circuit-QED platforms. Future work might extend the framework toward multivariate QSP for coupled anharmonic modes, explore efficient and coherent implementations without post-selections.

\section{Acknowledgments}
This work was partly supported by Basic Science Research Program through the National Research Foundation of Korea (NRF), funded by the Ministry of Science and ICT (RS-2023-NR068116, RS-2023-NR119931, RS-2025-03532992, RS-2025-07882969, RS-2024-00455131). This work was also partly supported by Institute for Information \& communications Technology Promotion (IITP) grant funded by the Korea government (MSIP) (No. 2019-0-00003, Research and Development of Core technologies for Programming, Running, Implementing and Validating of Fault-Tolerant Quantum Computing System). The Ministry of Trade, Industry, and Energy (MOTIE), Korea, also partly supported this research under the Industrial Innovation Infrastructure Development Project (Project No. RS-2024-00466693). JH is supported by the Yonsei University Research Fund of 2025-22-0140.

\section*{Appendix}
\appendix


\section{Values of molecular normal mode parameters}

Hamiltonian of the Uracil cation is built up based on the previous paper~\cite{Assmann2015, Vindel2022}. Notice that typos in the tables are corrected based on the potential energy surface plots. The parameters of potentials and the harmonic frequencies $\omega_r$ are given in $\mathrm{eV}$ and $\mathrm{cm}^{-1}$.

\begin{table}
\centering
\caption{Fitted Parameters: $\omega$ in $\text{cm}^{-1}$, $\kappa_i^{(\alpha)}$, $\lambda^{(\alpha\beta)}$ and $\gamma_{ii}^{(\alpha)}$ in eV for States $D_0$--$D_3$ ($\alpha = 0$--3) and Modes $\nu_3$, $\nu_7$, $\nu_{11}$, $\nu_{18}$--$\nu_{21}$ ($i = 3, 7, 11, 18$--21)~\cite{Assmann2015}. Parameters of the first order diagonal term $\sum_r \kappa^{(n)}_{r}\ket{n} \bra{n} \hat{Q}_r$, the first order off-diagonal term $ \sum_i \lambda^{(nm)}_i ( \ket{m}\bra{n}+ \ket{n}\bra{m}) \hat{Q}_i$ and the quadratic coupling term $\frac{1}{2}\sum_{i,n} \gamma^{(n)}_{i}\ket{n}\bra{n} \hat{Q}_i^2 $ are tabulated.}
\begin{tabular}{c|ccccccccccc}
\hline
Mode&$\omega$ & $\kappa^{(0)}$ & $\kappa^{(1)}$ & $\kappa^{(2)}$ & $\kappa^{(3)}$ & $\lambda^{(02)}$ & $\lambda^{(13)}$ & $\gamma^{(0)}$ & $\gamma^{(1)}$& $\gamma^{(2)}$& $\gamma^{(3)}$ \\
\hline
$\nu_3$  &388 & 0.04139 & -0.02688 & -0.05853 & 0.00132 &  &  & 0.00383 & -0.00426 & -0.00366& -0.00947 \\
$\nu_7$  &560 & -0.05367 & -0.00503 & 0.00775 & -0.04581 & 0.02141 &  &  &  & \\
$\nu_{11}$&770 & 0.05357 & -0.02456 & 0.03697 & -0.02130 & 0.02082 & 0.01845 &  & & &\\
$\nu_{18}$&1193 & -0.02203 & 0.09074 & 0.02748 & -0.04054 & -0.03538 & 0.08077 & 0.01938& 0.00694& -0.00294&0.00752 \\
$\nu_{19}$&1228 & 0.07472 & -0.00582 & 0.00889 & -0.02136 &-0.03763 & 0.05834 & 0.01590 & 0.01348& 0.00901 & 0.00497\\
$\nu_{20}$&1383 & -0.12147 & 0.05316 & 0.11233 & 0.00747 & -0.02049 &  & 0.01489 & 0.00828& 0.00183 & 0.00546\\
$\nu_{21}$&1406 & -0.09468 & 0.04454 & 0.14539 & 0.00050 &  & 0.07284 & 0.00970 & 0.00096& -0.00114 & 0.01108\\
\hline
\end{tabular}
\end{table}

\begin{table}
\centering
\caption{Quartic potentials $\frac{1}{24} k_r^{(a)} \hat{Q}_r^4$~\cite{Assmann2015}.}
\begin{tabular}{c|cccccc}
\hline
Mode & $\omega$&$k^{(0)}$ & $k^{(1)}$ & $k^{(2)}$ & $\lambda^{(01)}$ & $\lambda^{(12)}$ \\
\hline
$\nu_{10}$& 734& 0.03317 & 0.01157 & 0.01534 & 0.04633 & 0.03148 \\
$\nu_{12}$&771 & 0.02979 & 0.01488 & 0.01671 & 0.03540 & 0.03607 \\
\hline
\end{tabular}
\end{table}

\begin{table} 
\centering
\caption{Morse Potentials $V_k^{(\alpha)}(\hat{Q}_k) = d_k^{(\alpha)}[\mathrm{e}^{(a_k^{(\alpha)})(\hat{Q}_k-q_{k,0}^{(\alpha)})}-1]^2+e_{0}^{(\alpha)}$}
\begin{tabular}{c|r r r r r r r}
\hline
\textbf{Mode state} & $\omega$ & \multicolumn{1}{c}{\textbf{\( d_0 \)}} & \multicolumn{1}{c}{\textbf{\( a \)}} & \multicolumn{1}{c}{\textbf{\( q_0 \)}} & \multicolumn{1}{c}{\textbf{\( e_0 \)}} & \multicolumn{1}{c}{\textbf{\( \lambda^{(02)} \)}} & \multicolumn{1}{c}{\textbf{\( \lambda^{(13)} \)}} \\ \hline
\( \nu_{24} \)& 1673     &          &          &          &   &  & -0.01832     \\
\( D_0 \)&  & 41.89704 & -0.04719 & 0.81440 & -0.06431 & & \\ 
\( D_1 \)&  & 38.37122 & -0.05231 & 0.37488 & -0.01505 & & \\ 
\( D_2 \)&  & 39.25691 & -0.05286 & 0.14859 & -0.00244 & & \\ 
\( D_3 \)&  & 37.97847 & -0.05431 & -0.18152 & -0.00366 & & \\ \hline
\( \nu_{25} \) & 1761     &          &          &          &  &  0.00114 & 0.12606  \\
\( D_0 \)& & 4.80270 & 0.13675 & 0.02883 & -0.00007 &  &  \\ 
\( D_1 \)& & 74.15995 & 0.03064 & -1.34468 & -0.12082 & & \\ 
\( D_2 \)& & 90.76928 & 0.03374 & -0.29923 & -0.00916 & & \\ 
\( D_3 \)& & 20.56079 & 0.08044 & 0.38841 & -0.02071 & & \\ \hline
\( \nu_{26} \)& 1794     &          &          &          &  & 0.13035 & 0.14272 \\
\( D_0 \)& & 22.92802 & -0.07438 & -0.32069 & -0.01274 &  &  \\ 
\( D_1 \)& & 18.27440 & -0.07911 & -0.01711 & -0.00003 & & \\ 
\( D_2 \)& & 9.46894 & -0.08653 & 0.37635 & -0.01037 & & \\ 
\( D_3 \)& & 65.09678 & -0.03660 & 1.66312 & -0.25639 & & \\ 
\end{tabular}
\end{table}

\section{Failure of LVC/QVC approximations}

\begin{figure}
\centering
\includegraphics[width=0.95\textwidth]{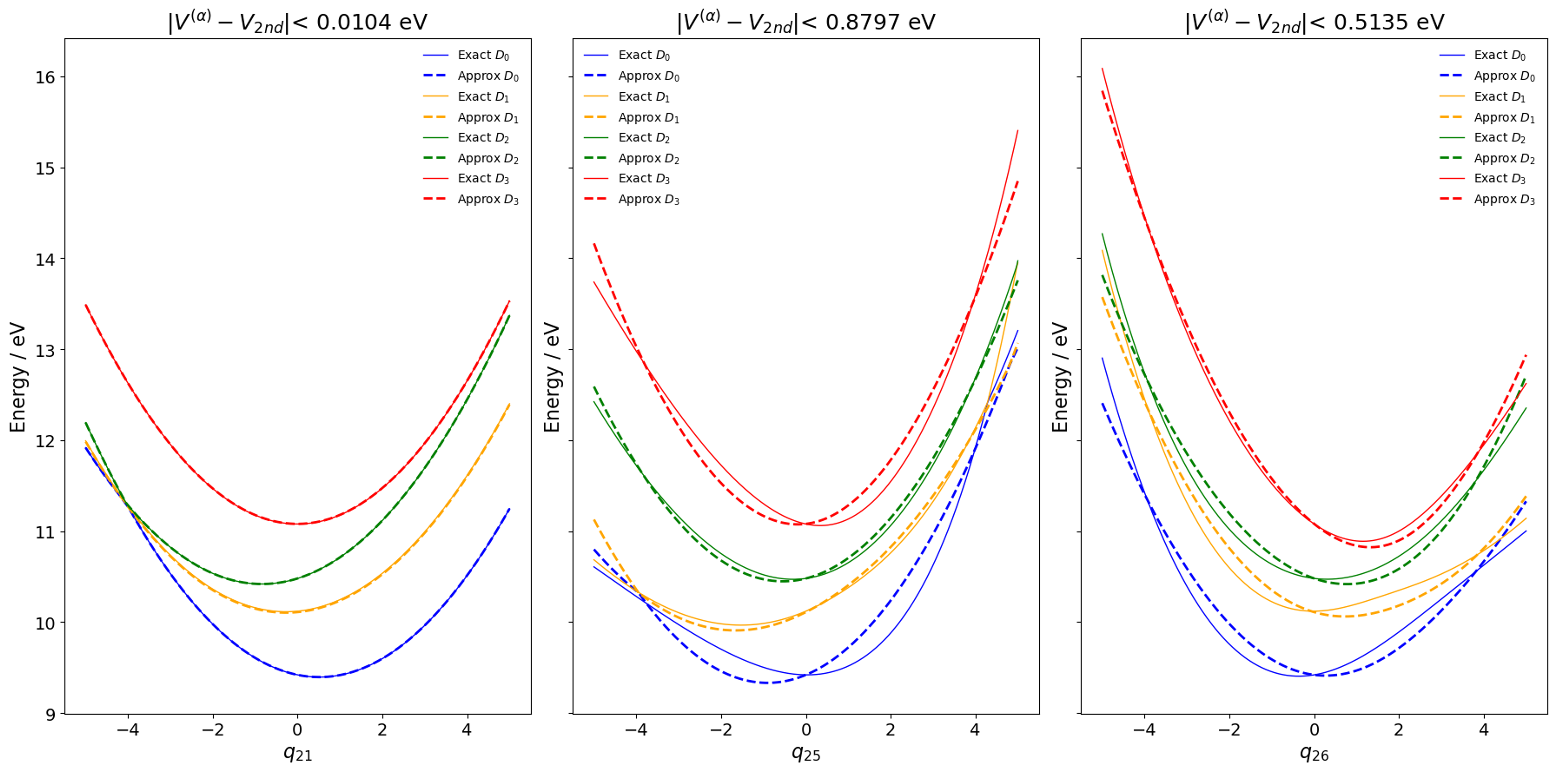}
\caption{Potential-energy cuts along 3 energy-dominant modes $\nu_{21}$, $\nu_{25}$ and $\nu_{26}$
coordinate. Even QVC can misplace the CIs by $>0.3$ eV.}
\label{fig:PES}
\end{figure}

The necessity of anharmonic terms becomes apparent when comparing harmonic models against the anharmonic model. Figure~\ref{fig:PES} reveals how LVC and quadratic vibronic coupling (QVC) approximations severely distort the potential energy surfaces along critical vibrational coordinates. These models misplace CIs along the $\nu_{25}$ and $\nu_{26}$ carbonyl stretch. These errors propagate to the dynamics. Figure~\ref{fig:MCTDH_dynamics} compares MCTDH population dynamics for initial excitation into the $D_2$ state. The anharmonic model predicts complete depopulation within 50 femtoseconds—consistent with experimental observations—while the LVC approximation shows virtually no decay. The QVC model performs marginally better but still fails to capture the rapid relaxation through the conical intersection. These failures stem from the harmonic models' inability to reproduce coordinate-dependent frequency shifts and mode mixing near the CI seam, validating the essential role of diagonal anharmonicity.

We compare the exact model with the QVC model, not with the LVC model, because the phase-space rotation gates $R_{nr}(\theta) = \exp[-\mathrm{i}\theta \hat{\sigma}_z \hat{a}_r^\dagger\hat{a}_r]$ enable native implementation upto quadratic potentials. Trapped-ion systems realize these through state-dependent squeezing, where driving at twice the motional frequency generates spin-dependent coordinate transformations in phase space~\cite{katz2023}. In circuit QED, the dispersive interactions between the bosonic mode and the auxiliary qubit realize a spin-dependent phase space rotation gate~\cite{Eickbusch2022}.

\begin{figure}
\centering
\includegraphics[width=0.95\textwidth]{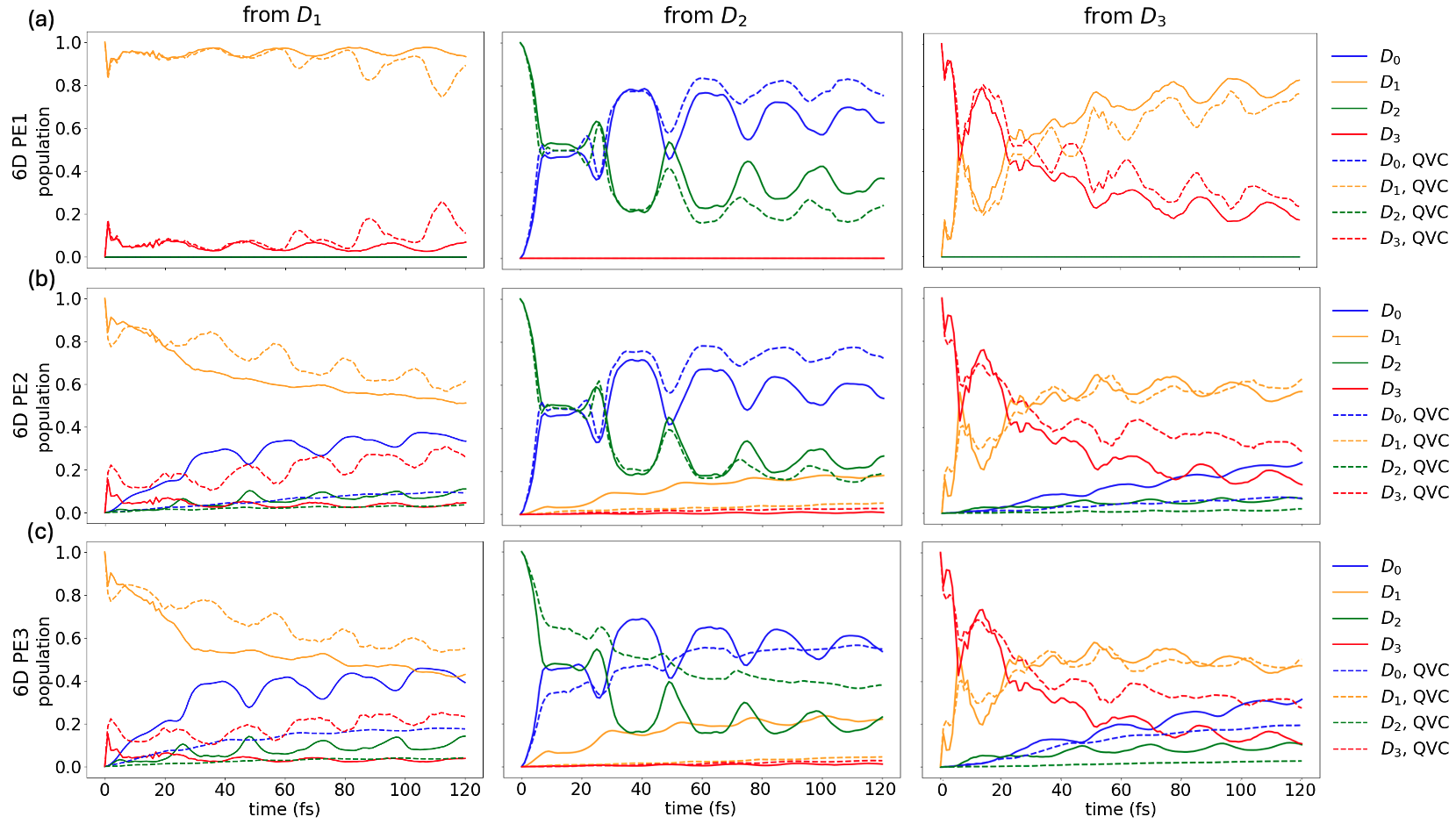}
\caption{MCTDH population dynamics for different mode set 6D PE1-3 and initial states $D_1\text{-}D_3$. Only the anharmonic model reproduces the ultrafast $\leq 50$ fs oscillating relaxation.}
\label{fig:MCTDH_dynamics}
\end{figure}

\section{Complexity analysis}

\begin{restatelemma}{lem:GQSP_qubit_oscillator_main}
Given a Z-basis controlled position displacement gate, one can construct complementary generalized signal operators
\begin{align}
\begin{aligned}
 A_{\hat{Q}} &= (\lvert 0 \rangle_q \langle 0 \rvert \otimes \hat{U} ) + (\lvert 1 \rangle_q \langle 1 \rvert \otimes \mathbb{I}_{\mathrm{osc}})\\ &= 
        \begin{bmatrix}
        \hat{U} & 0 \\
        0 & \mathbb{I}_{\mathrm{osc}}\\
        \end{bmatrix}, \\
 B_{\hat{Q}} &= (\lvert 0 \rangle_q \langle 0 \rvert \otimes \mathbb{I}_{\mathrm{osc}} ) + (\lvert 1 \rangle_q \langle 1 \rvert \otimes \hat{U}^{\dagger} )\\ &= 
        \begin{bmatrix}
        \mathbb{I}_{\mathrm{osc}}& 0 \\
        0 & \hat{U}^{\dagger}\\
        \end{bmatrix},\\ 
\end{aligned}
\end{align}
with one reusable ancilla qubit and 2 Z-basis CD gates, where $\hat{U} = \mathrm{e}^{\mathrm{i}\pi/L\hat{Q}}$. With these two signal operators, one can construct an arbitrary Fourier series $F(\hat{U})$ of the truncation order $d$ with the period $2L$ of the bosonic phase operator $\mathrm{e}^{\mathrm{i}\pi/L\hat{Q}}$ by Theorem~\ref{thm:GQSP}:
\begin{align}
\begin{aligned}
\operatorname{OQ-GQSP}:=&\mathrm{e}^{i\lambda\hat{\sigma}_z}\mathrm{e}^{i\phi_{-d}\hat{\sigma}_x} \mathrm{e}^{i\theta_{-d}\hat{\sigma}_z}\left( \prod_{r=-d+1}^{0} \hat{B}_{\hat{Q}} \mathrm{e}^{i\phi_{r}\hat{\sigma}_x} \mathrm{e}^{i\theta_{r}\hat{\sigma}_z}\right)\left( \prod_{r=1}^{d} \hat{A}_{\hat{Q}} \mathrm{e}^{i\phi_{r}\hat{\sigma}_x} \mathrm{e}^{i\theta_{r}\hat{\sigma}_z} \right) \\
=&\begin{pmatrix}
  F(\hat{U}) &-(G(\hat{U}))^{\dagger} \\
  G(\hat{U}) & (F(\hat{U}))^{\dagger} 
\end{pmatrix},
\end{aligned}
\end{align}
where $F(\hat{U}) = \sum\limits_{n=-d}^{d} c_n \mathrm{e}^{\mathrm{i}n\pi/L\hat{Q}}$ and $c_n$ is the Fourier series coefficient.
\end{restatelemma}
\begin{proof}
The proof is to make a signal operator $A_{\hat{Q}}  = 
\begin{bmatrix}
\mathrm{e}^{\mathrm{i}\frac{\pi}{L}\hat{Q}} & 0 \\
0 & I\\
\end{bmatrix}$ for GQSP in oscillator-qubit system from two CD gates: a position-displacement gate controlled by an ancilla qubit $\operatorname{CD}_{a}  = 
\begin{bmatrix}
\mathrm{e}^{\mathrm{i}\frac{\pi}{2L}\hat{Q}} & 0 \\
0 & \mathrm{e}^{-\mathrm{i}\frac{\pi}{2L}\hat{Q}}\\
\end{bmatrix}= \ket{0}_a\bra{0}\otimes \mathrm{e}^{\mathrm{i}\frac{\pi}{2L}\hat{Q}} + \ket{1}_a\bra{1}\otimes \mathrm{e}^{-\mathrm{i}\frac{\pi}{2L}\hat{Q}}$ and a position-displacement gate controlled by a main qubit $\operatorname{CD}_{q}  = 
\begin{bmatrix}
\mathrm{e}^{\mathrm{i}\frac{\pi}{2L}\hat{Q}} & 0 \\
0 & \mathrm{e}^{-\mathrm{i}\frac{\pi}{2L}\hat{Q}}\\
\end{bmatrix}= \ket{0}_q\bra{0}\otimes \mathrm{e}^{\mathrm{i}\frac{\pi}{2L}\hat{Q}} + \ket{1}_q\bra{1}\otimes \mathrm{e}^{-\mathrm{i}\frac{\pi}{2L}\hat{Q}}$. They are conventional the signal operators of the Z-basis hybrid oscillator-qubit QSP~\cite{liu2024toward}. 

Let $U = \mathrm{e}^{\mathrm{i}\frac{\pi}{2L} \hat{Q}}$ and $V = \mathrm{e}^{\mathrm{i} \frac{\pi}{L} \hat{Q}} = U^2$. We show how to implement the operator $A_{\hat{Q}} = \ket{0}_m\bra{0}\otimes V + \ket{1}_m\bra{1}\otimes \mathbb{I}_{osc}$ on the main qubit ($m$) and oscillator system, using one ancilla qubit ($a$) prepared in state $\ket{0}_a$ concatenating two CD gates.
Consider the sequence $\operatorname{CD}_m \operatorname{CD}_a$ acting on an initial state $(\alpha_m\ket{0}_m+\beta_m\ket{1}_m)\ket{0}_a\ket{\mathrm{osc}}$:
\begin{align}
  \begin{aligned}
&(\alpha_m\ket{0}_m+\beta_m\ket{1}_m)\ket{0}_a\ket{\mathrm{osc}} \\
\xrightarrow{\operatorname{CD}_a} & (\alpha_m\ket{0}_m+\beta_m\ket{1}_m) \left[ (\ket{0}_a\bra{0}\otimes U + \ket{1}_a\bra{1}\otimes U^{-1}) \ket{0}_a\ket{\mathrm{osc}} \right] \\
&= (\alpha_m\ket{0}_m+\beta_m\ket{1}_m) \ket{0}_a (U\ket{\mathrm{osc}})\\
\xrightarrow{\operatorname{CD}_m} & (\ket{0}_m\bra{0}\otimes U + \ket{1}_m\bra{1}\otimes U^{-1}) \left[ \alpha_m\ket{0}_m\ket{0}_a (U\ket{\mathrm{osc}}) + \beta_m\ket{1}_m\ket{0}_a (U\ket{\mathrm{osc}}) \right] \\
&= \alpha_m\ket{0}_m\ket{0}_a (U \cdot U\ket{\mathrm{osc}}) + \beta_m\ket{1}_m\ket{0}_a (U^{-1} \cdot U\ket{\mathrm{osc}}) \\
&= \alpha_m\ket{0}_m\ket{0}_a (U^2\ket{\mathrm{osc}}) + \beta_m\ket{1}_m\ket{0}_a (\mathbb{I}_{osc}\ket{\mathrm{osc}}) \\
&= \left[ (\alpha_m\ket{0}_m \otimes V + \beta_m\ket{1}_m \otimes \mathbb{I}_{osc}) \ket{\mathrm{osc}} \right] \otimes \ket{0}_a.
\end{aligned}
\end{align}
This sequence implements the operator $A_{\hat{Q}} = \text{diag}(V, \mathbb{I}_{osc}) = \text{diag}(U^2, \mathbb{I}_{osc})$ on the main oscillator-qubit subspace, while preserving the ancilla state $\ket{0}_a$, allowing its reuse.
\end{proof}

\begin{restatelemma}{lem:GQSP_qubit_oscillator_state_dependent}
Given a state in oscillator-qubit processor $\ket{\psi_D}\ket{\mathrm{osc}}\ket{0}_a$ with four inverted-unary‑encoded electronic states $\ket{\psi_{D}} = a_0\ket{D_0}+a_1\ket{D_1}+a_2\ket{D_2}+a_3\ket{D_3}$, an oscillator $\ket{\mathrm{osc}}$, and an ancilla qubit $\ket{0}_a$, we can implement a state-dependent nonlinear bosonic phase gate
$\mathrm{e}^{i\ket{D_n}\bra{D_n}V_{nr}(\hat{Q}_{r})}$
with two OQ-GQSP operators in Lemma~\ref{lem:GQSP_qubit_oscillator_main} and 2 measurements with which each success probability is $1-\delta$ where $\delta = \lvert | G(\hat{U})  \rvert |^2 \ll 1.$
\end{restatelemma}
\begin{proof}
Let \begin{align}
\begin{aligned}
\operatorname{OQ-GQSP}_a=&\begin{pmatrix}
F(\hat{U}) &-(G(\hat{U}))^{\dagger} \\
G(\hat{U}) & (F(\hat{U}))^{\dagger} 
\end{pmatrix}_a\\  =& \ket{0}_a\bra{0}F(\hat{U}) - \ket{0}_a\bra{1}(G(\hat{U}))^{\dagger} +  \ket{1}_a\bra{0}G(\hat{U}) + \ket{1}_a\bra{1}(F(\hat{U}))^{\dagger} 
\end{aligned}
\end{align}
and
\begin{align}
\begin{aligned}
\operatorname{OQ-GQSP}_m=&\begin{pmatrix}
F(\hat{U}) &-(G(\hat{U}))^{\dagger} \\
G(\hat{U}) & (F(\hat{U}))^{\dagger} 
\end{pmatrix}_m\\ =& \ket{0}_m\bra{0}F(\hat{U}) -  \ket{0}_m\bra{1}(G(\hat{U}))^{\dagger} +  \ket{1}_m\bra{0}G(\hat{U}) + \ket{1}_m\bra{1}(F(\hat{U}))^{\dagger}.
\end{aligned}
\end{align}
\begin{align}
\begin{aligned}
\ket{\psi_D}\ket{\mathrm{osc}}\ket{0}_a  \xrightarrow{\operatorname{OQ-GQSP}_a}  & \ket{\psi_D}\left( F(\hat{U})\ket{\mathrm{osc}}\ket{0}_a + G(\hat{U})\ket{\mathrm{osc}}\ket{1}_a \right)\\
\xrightarrow{\operatorname{Measure }{\ket{0}_a\bra{0}}} &\ket{\psi_D}F(\hat{U})\ket{\mathrm{osc}}\ket{0}_a  \text{ with probability $\boldsymbol{1-\delta}$} \\
\xrightarrow{\operatorname{OQ-GQSP}_m} &a_0\left( \ket{1110}F(\hat{U})^2\ket{\mathrm{osc}}\ket{0}_a  +   \ket{1111}G(\hat{U})F(\hat{U})\ket{\mathrm{osc}}\ket{0}_a \right)\\
&a_1 \left( \ket{1100}(-G(\hat{U})^{\dagger})F(\hat{U})\ket{\mathrm{osc}}\ket{0}_a  +   \ket{1101} F(\hat{U})^{\dagger}F(\hat{U})\ket{\mathrm{osc}}\ket{0}_a \right)\\
&a_2 \left( \ket{1010}(-G(\hat{U})^{\dagger})F(\hat{U})\ket{\mathrm{osc}}\ket{0}_a  +   \ket{1011} F(\hat{U})^{\dagger}F(\hat{U})\ket{\mathrm{osc}}\ket{0}_a \right)\\
&a_3 \left( \ket{0110}(-G(\hat{U})^{\dagger})F(\hat{U})\ket{\mathrm{osc}}\ket{0}_a  +   \ket{0111} F(\hat{U})^{\dagger}F(\hat{U})\ket{\mathrm{osc}}\ket{0}_a \right)\\
\xrightarrow{\operatorname{Parity\, measure}} &a_0\left( \ket{1110}F(\hat{U})^2\ket{\mathrm{osc}}\ket{1}_a  +   \ket{1111}G(\hat{U})F(\hat{U})\ket{\mathrm{osc}}\ket{0}_a \right)\\
&a_1 \left( \ket{1100}(-G(\hat{U})^{\dagger})F(\hat{U})\ket{\mathrm{osc}}\ket{0}_a  +   \ket{1101} F(\hat{U})^{\dagger}F(\hat{U})\ket{\mathrm{osc}}\ket{1}_a \right)\\
&a_2 \left( \ket{1010}(-G(\hat{U})^{\dagger})F(\hat{U})\ket{\mathrm{osc}}\ket{0}_a  +   \ket{1011} F(\hat{U})^{\dagger}F(\hat{U})\ket{\mathrm{osc}}\ket{1}_a \right)\\
&a_3 \left( \ket{0110}(-G(\hat{U})^{\dagger})F(\hat{U})\ket{\mathrm{osc}}\ket{0}_a  +   \ket{0111} F(\hat{U})^{\dagger}F(\hat{U})\ket{\mathrm{osc}}\ket{1}_a \right)\\
\xrightarrow{\operatorname{Measure }{\ket{1}_a\bra{1}}} & a_0 \ket{1110}F(\hat{U})^2\ket{\mathrm{osc}}\ket{1}_a + a_1 \ket{1101}F(\hat{U})^{\dagger}F(\hat{U})\ket{\mathrm{osc}}\ket{1}_a \\
+&a_2\ket{1011}F(\hat{U})^{\dagger}F(\hat{U})\ket{\mathrm{osc}}\ket{1}_a + a_3 \ket{0111}F(\hat{U})^{\dagger}F(\hat{U})\ket{\mathrm{osc}}\ket{1}_a\\ &\text{ with probability $\boldsymbol{1-\delta}$}\\
\approx& a_0 \ket{1110}F(\hat{U})^2\ket{\mathrm{osc}}\ket{1}_a + a_1 \ket{1101}\ket{\mathrm{osc}}\ket{1}_a \\
+&a_2\ket{1011}\ket{\mathrm{osc}}\ket{1}_a + a_3 \ket{0111}\ket{\mathrm{osc}}\ket{1}_a\\
\xrightarrow{\operatorname{Reset \, ancilla}{\ket{0}_a}}& a_0 \ket{1110}F(\hat{U})^2\ket{\mathrm{osc}}\ket{0}_a + a_1 \ket{1101}\ket{\mathrm{osc}}\ket{0}_a \\
+&a_2\ket{1011}\ket{\mathrm{osc}}\ket{0}_a + a_3 \ket{0111}\ket{\mathrm{osc}}\ket{0}_a ,
\end{aligned}
\end{align}
where the first measurement on the ancilla qubit yields implementation of $F(\hat{U})$ operator on the oscillator state and the second measurement of the parity on the main qubit yields implementation of $F(\hat{U})^2$ on the oscillator state coupled with the target electronic state. This implements the desired state-dependent phase gate with heralded success of $(1-\delta)^2$ without collapsing the electronic state data of the main qubits.
\end{proof}

OQ-GQSP approximates nonlinear bosonic phase gates via truncated Fourier series. Before analyzing the complexity of simulating the vibronic coupling Hamiltonian with this method, we first derive the function approximation error, based on the fact that the exponential decay of Fourier coefficients is equivalent to the analyticity of the target function.

For this analysis, we assume the simulation is restricted to a compact real interval $x \in [-L, L]$ with $L \in \mathbb{R}$. Hence, reconstructed potential energy surfaces of uracil cation are expanded in the orthonormal Fourier basis $\{e^{\,\mathrm{i}\pi k x / L}\}_{k\in\mathbb Z}$. The relevant Fourier series takes the form $\sum_k c_k \mathrm{e}^{\mathrm{i}\pi k x / L}$, which can be viewed as a special case of a Laurent series $\sum_k c_k w^k$ under the change of variables $w = \mathrm{e}^{\mathrm{i}\pi x/L}$, where $|w| = 1$. This identification allows us to apply complex-analytic estimates.

\begin{lemma}[Fourier truncation for phase functions]\label{lem:trunc-L}
Let $L>0$, $\sigma>0$, and $\varepsilon\in(0,1)$.  
Suppose $f:[-L,L]\to\mathbb R$ is real‑analytic and set $g(x):=e^{\mathrm i \Delta t f(x)}$
extended $2L$‑periodically to $\mathbb R$.
Assume $g$ admits a holomorphic extension to the strip
$\mathcal S_\sigma:=\{z\in\mathbb C:\ |\Im z|<\sigma\}$ with
$\displaystyle \sup_{z\in\mathcal S_\sigma}\lvert g(z)\rvert \le B$ for some $B\ge 1$.

\smallskip
Write its Fourier series
\begin{align}
\begin{aligned}
  g(x)=\sum_{k\in\mathbb Z} a_k\,e^{\mathrm i\frac{\pi k}{L}x},
  \qquad
  a_k=\frac1{2L}\int_{-L}^{L} g(x)\,e^{-\mathrm i\frac{\pi k}{L}x}\,dx .
\end{aligned}
\end{align}

\noindent
Then, with \(\rho:=e^{\pi\sigma/L}>1\),
\begin{align}
\begin{aligned}
  \;|a_k|\le B\,\rho^{-|k|}\quad(k\in\mathbb Z) .
\end{aligned}
\end{align}

\noindent
Hence the $d$‑term truncation  
\(g_{(d)}(x):=\sum\limits_{|k|\le d}a_k e^{\mathrm i\frac{\pi k}{L}x}\) obeys
\begin{align}
\begin{aligned}
  \|g-g_{(d)}\|_{[-L,L]}
  \le \frac{2B\,\rho^{-(d+1)}}{1-\rho^{-1}}
  \le \varepsilon
\end{aligned}
\end{align}
provided
\begin{align}
\begin{aligned}
  d\;\ge\;
  \left\lceil
     \frac{L}{\pi\sigma}\,
     \ln\!\Bigl(\tfrac{2B}{(1-\rho^{-1})\varepsilon}\Bigr)
  \right\rceil
  =\;
  \mathcal O\!\bigl(\ln\varepsilon^{-1}\bigr)\text{ for fixed }\sigma.
\end{aligned}
\end{align}

\end{lemma}

\begin{proof}
Put $w:=e^{\mathrm i\pi z/L}$; this maps $\mathcal S_\sigma$ conformally
onto the annulus $\mathcal A_\rho:=\{\rho^{-1}<|w|<\rho\}$ with
$\rho=e^{\pi\sigma/L}$.  
Set $\tilde g(w):=g\!\bigl(\tfrac{L}{\mathrm i\pi}\log w\bigr)$,
using the principal branch of $\log$.  
Then $\tilde g$ is analytic on $\mathcal A_\rho$ and
$|\tilde g(w)|\le B$ there.

By change of variables $w=e^{\mathrm i\pi x/L}$ we have
$dx=\tfrac{L}{\mathrm i\pi}\,w^{-1}dw$ and
\begin{align}
\begin{aligned}
  a_k
  =\frac{1}{2\pi\mathrm i}\,
    \oint_{|w|=1}\tilde g(w)\,w^{-k-1}\,dw .
\end{aligned}
\end{align}

Deforming the contour to $|w|=\eta$ with $\rho^{-1}<\eta<\rho$ yields the Cauchy estimate
\begin{align}
\begin{aligned}
  |a_k|
  \le B\,\eta^{-|k|}\qquad(k\in\mathbb Z).
\end{aligned}
\end{align}
Taking the limit $\eta\to \rho$ for $k\ge 0$ and $\eta\to \rho^{-1}$ for $k\le 0$ then implies 
\begin{align}
\begin{aligned}
|a_k|\le B\,\rho^{-|k|}.
\end{aligned}
\end{align}

Summing the geometric tails,
\begin{align}
\begin{aligned}
  \sum_{|k|>d}|a_k|
  \le
  2B\sum_{k=d+1}^\infty \rho^{-k}
  =\frac{2B\,\rho^{-(d+1)}}{1-\rho^{-1}},
\end{aligned}
\end{align}
which establishes the stated bound on $d$.
\end{proof}

For physically relevant potentials that are entire, e.g., polynomials, Morse, trigonometric functions, we may choose $\sigma=L$, which gives $\rho=e^{\pi}\approx 23.14$. In this regime, the truncation degree simplifies to $d=\mathcal O(\ln(1/\varepsilon))$ up to absolute constants. If $1<\rho<2$, retain the prefactor $(1-\rho^{-1})^{-1}$ in the tail bound; this does not arise under $\sigma=L$.

\begin{restatetheorem}{thm:GQSP_VCdynamics}
Let $N$ be the number of electronic states, $M$ the number of vibrational modes, and $M'$ the number of anharmonic modes. For
\begin{align}
  \hat{H}_{\mathrm{VC}}
  =\sum_{n=0}^{N-1}\hat H_{\mathrm{diag}}^{(n)}
    +\!\!\sum_{\substack{n,m=0 \\ n \ne m}}^{N-1} \hat H_{\mathrm{off}}^{(n,m)} ,
\end{align}
set
\begin{align}
  \hat H_{\mathrm{diag}}^{(n)}
  =\sum_{r=0}^{M-1}\ket n\!\bra n\otimes f_{nr}(\hat Q_r),\qquad
  \hat H_{\mathrm{off}}^{(n,m)}
  =(\ket n\!\bra m)
    \otimes\sum_{r=0}^{M-1}\lambda_{nmr}\hat Q_r .
\end{align}

Fix $L>0$. For each pair $(n,r)$ set $g_{nr}(x):=\exp\!\bigl(\mathrm i\,\Delta t f_{nr}(x)\bigr)$. Assume the analyticity and boundedness hypotheses of Lemma~\ref{lem:trunc-L} for each $g_{nr}(x)$.

Let $\varepsilon\in(0,1)$ denote the total target error and split it equally between product–formula and Fourier–truncation contributions:
\begin{align}
  \varepsilon=\varepsilon_{\mathrm{Trot}}+\varepsilon_{\mathrm{Fourier}},
  \qquad
  \varepsilon_{\mathrm{Trot}}=\varepsilon_{\mathrm{Fourier}}=\varepsilon/2.
\end{align}

Let $\Delta t\in(0,1]$ and $\varepsilon_{\mathrm{Trot}}\in(0,1)$ and choose $\Delta t$ so that
\begin{align}
  \frac{\Delta t^{2}}{2}\,
  \Bigl\|\Bigl[\sum_{n}\hat H_{\mathrm{diag}}^{(n)}\!,\ \sum_{n \ne m}\hat H_{\mathrm{off}}^{(n,m)}\Bigr]\Bigr\|
  \le \varepsilon_{\mathrm{Trot}} .
\end{align}
Then OQ–GQSP yields a unitary
\begin{align}
  \tilde U
  =\prod_{n=0}^{N-1}\mathrm{e}^{-i\hat H_{\mathrm{diag}}^{(n)}\Delta t}
    \prod_{\substack{n,m=0 \\ n \ne m}}^{N-1}\mathrm{e}^{-i\hat H_{\mathrm{off}}^{(n,m)}\Delta t}
\end{align}
such that $\bigl\|\mathrm{e}^{-i\hat H_{\mathrm{VC}}\Delta_t}-\tilde U\bigr\|
\le \varepsilon$,
using $\mathcal{O}\!\bigl(N M'\ln\varepsilon^{-1}+ N^2 M\bigr)$
queries to $\mathrm{CD}$ gates and achieving success probability
$\,(1-\delta)^{2M'N}$.
\end{restatetheorem}

\begin{proof}
First‑order splitting gives
\begin{align}
  \bigl\|e^{-i\hat H_{\mathrm{VC}}\Delta t}
  -e^{-i\sum_{n}\hat H_{\mathrm{diag}}^{(n)}\Delta t}\,
   e^{-i\sum_{n<m}\hat H_{\mathrm{off}}^{(n,m)}\Delta t}\bigr\|
  \le \frac{\Delta t^{2}}{2}\,
  \Bigl\|\Bigl[\sum_{n}\hat H_{\mathrm{diag}}^{(n)}\!,\ \sum_{n<m}\hat H_{\mathrm{off}}^{(n,m)}\Bigr]\Bigr\| ,
\end{align}
so the chosen $\Delta t$ bounds the BCH error by $\varepsilon_{\mathrm{Trot}}$.
Lemma~\ref{lem:GQSP_qubit_oscillator_main} converts
$g_{nr}(x)=\exp\!\bigl(i\Delta t f_{nr}(x)\bigr)$ into a length‑$d$
Fourier series
\begin{align}
g_{nr}(x)=\sum_{k=-d}^{d}c^{(nr)}_{k}\,e^{\,i\frac{\pi k}{L}x}
\;+\;R_{nr}^{(d)}(x),\qquad
\|R_{nr}^{(d)}\|_{\infty}\le\varepsilon_{\mathrm{Fourier}},
\end{align}
and Lemma~\ref{lem:GQSP_qubit_oscillator_state_dependent} implements
$\mathrm{e}^{-i\hat H_{\mathrm{diag}}^{(n,r)}\Delta t}
=\ket n\!\bra n\otimes g_{nr}(\hat Q_r)$
with $\mathcal{O}(d)$ calls to $\mathrm{CD}$.
The diagonal sector therefore costs
$\mathcal{O}\!\bigl(NM'd+N(M-M')\bigr)$ queries, since linear vibronic
diagonal modes use $\mathcal{O}(M-M')$ queries.
Choosing $d=\mathcal{O}\!\bigl(\ln\varepsilon_{\mathrm{Fourier}}^{-1}\bigr)$
by Lemma~\ref{lem:trunc-L} for $NM$ diagonal terms gives
$\mathcal{O}\!\bigl(NM'\ln\varepsilon_{\mathrm{Fourier}}^{-1}+N(M-M')\bigr)$
queries in that sector. The $\ln\varepsilon^{-1}$ factor is explicitly $\ln\varepsilon_{\mathrm{Fourier}}^{-1}$. The off‑diagonal terms are already linear, so no
Fourier expansion is needed; their cost is $\frac{N(N-1)}{2}M\times 2=\mathcal{O}(N^{2}M)$ queries. Adding both contributions yields the stated gate count. The extra $\mathcal{O}\!\bigl(N(M-M')\bigr)$ term is asymptotically dominated by $\mathcal{O}(N^{2}M)$.
\end{proof}


\clearpage
\bibliography{mybibfile.bib}

\end{document}